\documentclass[11pt,journal,onecolumn,draft]{IEEEtran}

\usepackage{hyperref}
\usepackage{multirow}
\usepackage{amsmath,amsthm}
\usepackage{amsfonts,amssymb}
\usepackage{bm}
\usepackage{xspace}
\usepackage{enumitem}
\usepackage{dsfont}
\usepackage{marginnote}
\usepackage[space,nosort]{cite} % To insert a space between multiple citations
\usepackage{booktabs}

\newtheorem{definition}{Definition}
\newtheorem{theorem}{Theorem}

\newtheorem{proposition}[theorem]{Proposition}

\newtheorem{remark}[theorem]{Remark}
\newtheorem{assumption}{Assumption}

\newcommand{\f}{\mathbb F}

\newcommand{\x}{\bf x}

\newcommand{\e}{\bf e}

\newcommand{\Ee}{\mathcal{E}}
\newcommand{\Ae}{\mathcal{A}}
\newcommand{\Setup}{\normalfont\textsf{Setup}}
\newcommand{\KeyGen}{\normalfont\textsf{KeyGen}}
\newcommand{\Encrypt}{\normalfont\textsf{Encrypt}}
\newcommand{\Decrypt}{\normalfont\textsf{Decrypt}}
\newcommand{\seck}{\lambda}
\newcommand{\param}{\normalfont\textsf{param}}
\newcommand{\pk}{\textsf{pk}}
\newcommand{\sk}{\textsf{sk}}
\newcommand{\INDCPA}{\ensuremath{\normalfont\textsf{IND\text{-}CPA}}\xspace}

\newcommand{\Adv}{\mathbf{Adv}}
\newcommand{\Exp}{\mathbf{Exp}}
\newcommand{\ind}{\normalfont\textsf{ind}}
\newcommand{\sets}{\leftarrow}
\newcommand{\FIND}{\texttt{FIND}}
\newcommand{\GUESS}{\texttt{GUESS}}
\newcommand{\comreturn}{\texttt{RETURN\ }}
\newcommand{\DQCSD}{\textsf{DQCSD}\xspace}
\newcommand{\QCSD}{\textsf{QCSD}\xspace}
\newcommand{\SD}{\textsf{SD}\xspace}
\newcommand{\DSD}{\textsf{DSD}\xspace}
\newcommand{\RSD}{\textsf{RSD}\xspace}
\newcommand{\DRSD}{\textsf{DRSD}\xspace}
\newcommand{\RQCSD}{\textsf{RQCSD}\xspace}
\newcommand{\DRQCSD}{\textsf{DRQCSD}\xspace}

\usepackage{games}
\newcommand{\mathgras}[1]{\ensuremath{%
		{\text{\mathversion{bold}\ensuremath{#1}}}%
	}}
	\newcommand{\bG}{$\mathgras{G}$}

\usepackage[utf8]{inputenc}
\usepackage[english]{babel}
\usepackage{color}
\usepackage{graphicx}
\usepackage{multirow}

\newcommand\carlos[1]{\textbf{\color{magenta}{\bf Carlos:} #1\color{black}}}
\newcommand\jc[1]{\textbf{\color{vert}{\bf JC:} #1\color{black}}}
\newcommand\philippe[1]{\textbf{\color{orange}{\bf Philippe:} #1\color{black}}}
\newcommand\ob[1]{\textbf{\color{blue}{\bf OB:} #1\color{black}}}

\iftrue
\renewcommand\carlos[1]{}
\renewcommand\jc[1]{}
\renewcommand\philippe[1]{}
\renewcommand\ob[1]{}
\renewcommand\marginnote[1]{}
\fi
\definecolor{vert}{rgb}{0,0.546875,0}

\newcommand{\lar}{\stackrel{\mathdollar}{\leftarrow}}

\definecolor{orange}{rgb}{1,.5,.0}

\newcommand\ul[1]{\textcolor{blue}{#1}}

\hypersetup{
    colorlinks=true, % Add color to links
    breaklinks=true,  % break lines for longest links
    urlcolor= blue,   % URL hyperlinks' color [blue]
    linkcolor= blue,  % Intern hyperlinks' color (index, figures, tables, equations, ...) [blue]
    citecolor= vert  % Biblio hyperlinks' color [green]
    }

\let\oldnl\nl% Store \nl in \oldnl
\newcommand{\nonl}{\renewcommand{\nl}{\let\nl\oldnl}}% Remove line number for one line

\begin{document}

\title{Efficient Encryption from Random Quasi-Cyclic Codes}

\author{Carlos Aguilar, Olivier Blazy, Jean-Christophe Deneuville,\\
  Philippe Gaborit and Gilles Z\'emor}
%\institute{}
% \author{C. Aguilar\inst{1}, P. Gaborit\inst{1},  \and G. Z\'emor\inst{2}}
% \date{}

% \institute{  Universit\'e de Limoges, XLIM-DMI, \\
% 123, Av. Albert Thomas \\
% 87060 Limoges Cedex, France.\\
%\email{carlos.aguilar,gaborit,patrick.lacharme,julien.schrek@unilim.fr}
% \and
% Institut de Math\'ematiques de Bordeaux- Universit\'e Bordeaux,\\
% 351 cours de la Lib\'eration, 33400, Talence, France.\\
% \email{gilles.zemor@math.u-bordeaux.fr}
%}
\sloppy
\maketitle
\pagestyle{plain}
%
%\begin{center}
%\begin{itemize}
%\item Contributions ?
%\item Related works ?
%\item $p_{fail}$ pour previous constructions ?
%\item Références ?
%\item Assumption 1 : philippe : ``papier JPT''
%\end{itemize}
%\end{center}

\begin{abstract}
We propose a framework for constructing efficient code-based encryption schemes
from codes that do not hide any structure in their public matrix.
The framework is in the spirit of the schemes first proposed by
Alekhnovich in 2003 and based on the difficulty of decoding random
linear codes from random errors of low weight. We depart somewhat from Aleknovich's
approach and propose an encryption scheme based on the difficulty of
decoding random quasi-cyclic codes.
We propose two new cryptosystems instantiated within our framework: the
Hamming Quasi-Cyclic cryptosystem (HQC), based on the Hamming metric, and
the Rank Quasi-Cyclic cryptosystem (RQC), based on the rank metric. 
We give a security proof, which reduces the IND-CPA security of our
systems to a decisional version of the well known
problem of  decoding random families of quasi-cyclic codes for the Hamming  
and rank metrics (the respective QCSD and RQCSD problems).
 We also provide an analysis of the decryption 
failure probability of our scheme in the Hamming metric case: for the
rank metric there is no decryption failure.
Our schemes benefit from a very fast decryption algorithm together with small key sizes of only 
a few thousand bits. The cryptosystems are very efficient for low
encryption rates and are very 
well suited to key exchange and authentication. 
Asymptotically, for $\lambda$ the security parameter, the public key
sizes are respectively
in $\mathcal{O}({\lambda}^2)$ for HQC and in $\mathcal{O}(\lambda^{\frac{4}{3}})$ for RQC.
Practical parameter compares well to systems based on ring-LPN or the recent MDPC system. 
\end{abstract}

\vspace{1.25cm}

%\vspace{1.25cm}
\section{Introduction}
\subsection{Background and Motivation}
The first code-based cryptosystem was proposed by McEliece in 1978. This system, which can be seen as a general encryption setting for coding theory,
is based on a hidden trapdoor associated to a decodable
family of codes, hence a strongly structured family of codes. The inherent construction of the system makes it difficult to formally reduce security to the generic difficulty of
decoding random codes. Even if the original McEliece cryptosystem, based on the family of Goppa codes, is still considered secure today, many variants
based on alternative families of codes (Reed-Solomon codes, Reed-Muller codes or some alternant codes \cite{SAC:MisBar09,AFRICACRYPT:BCGO09}) were broken by recovering in polynomial time the hidden structure \cite{EC:FOPT10}. 
The fact that the hidden code structure may be uncovered (even possibly for Goppa codes) lies like a sword of Damocles over the system, and
finding a practical alternative cryptosystem based on the difficulty of decoding unstructured
or random codes has always been a major issue in code-based cryptography.
The recently proposed MDPC cryptosystem \cite{MTSB13} (somewhat in the spirit of the NTRU cryptosystem \cite{HPS98}) addresses the problem by using a hidden code structure which is significantly weaker than
that of previously used algebraic codes like Goppa codes. The cryptosystem
\cite{GMRZ13} followed this trend with a similar approach.
Beside this weak hidden structure, the MDPC system has very nice features and in particular relatively small key sizes,
because of the cyclic structure of the public matrix. However,
even if this system is a strong step forward for code-based cryptography, the hidden structure issue has not altogether disappeared. 

In 2003, Alekhnovich proposed an innovative approach based on the difficulty of decoding
purely random codes~\cite{FOCS:Alekhnovich03}. In this system the trapdoor (or secret key) is a random error vector that has been added to a random codeword of a random code.
Recovering the secret key is therefore equivalent to solving the problem of
decoding a random code -- with no hidden structure. Alekhnovich also proved that breaking the system in any way, not necessarily by recovering the secret key, involves decoding a random linear code.

Even if the system was not totally practical, the approach in itself
was a breakthrough for code-based cryptography. Its inspiration was
provided in part by the Ajtai-Dwork cryptosystem \cite{AD97} which is based on solving hard lattice problems. 
The Ajtai-Dwork cryptosystem also inspired the Learning With Errors (LWE) lattice-based cryptosystem by Regev \cite{STOC:Regev03} which generated a huge amount of work in lattice-based cryptography. Attempts to emulate this approach in code-based cryptography were also made and systems based on the Learning Parity with Noise (LPN) have been proposed by exploiting the analogy with LWE \cite{DV13,PKC:KilMasPie14}: the LPN problem is essentially the problem of decoding random linear codes of fixed dimension and unspecified length over a binary symmetric channel. 
The first version of the LWE cryptosystem was not very efficient, but
introducing more structure in the public
key (as for NTRU) lead to the very efficient Ring-LWE cryptosystem \cite{EC:LyuPeiReg10}. One strong feature of this last paper is that 
it gives a reduction from the decisional version of the ring-LWE problem to a search version of the problem. Such a reduction is not known for the case of the ring-LPN problem. A ring version (ring-LPN) was nevertheless introduced in \cite{Lapin12} for authentication and
for encryption in \cite{damgaard2012public}. 

In this paper, we propose an efficient cryptosystem based on the difficulty of decoding random quasi-cyclic codes. It is inspired by Ring-LWE encryption but is significantly adapted to the coding theory setting.
Our construction benefits from some nice features: a reduction to a decisional version of the general problem of decoding random quasi-cyclic codes, hence 
with no hidden structure, and also quite good parameters and efficiency. 
Since our approach is relatively general, it can also
be used with other metrics such as the rank metric. Finally, another strong feature of our approach is that inherently it leads to a precise analysis of the decryption failure probability, which is also a hard point for the MDPC cryptosystem
and is not done in detail for other approaches based on the LPN problem.
A relative weakness of our system is its relatively low encryption rate, but this is not a major
issue for classical applications of public-key encryption schemes such as authentication or key exchange.

\subsection{Our Contributions}

We propose the first efficient code-based cryptosystem whose security relies on decoding small weight vectors 
of random quasi-cyclic codes. We provide a reduction of our cryptosystem to this problem together with a detailed analysis 
of the decryption failure probability. Our analysis allows us to give small parameters for code-based encryption in Hamming and Rank metrics. 
When compared to the MDPC~\cite{MTSB13} or LRPC~\cite{GMRZ13}
cryptosystems, our proposal offers higher security (in terms of
  security bits) and better decryption guarantees for similar
parameters (i.e. key and communication size), but with a lower encryption rate. Overall we propose concrete parameters for different 
levels of security, in both the classical and quantum settings. These parameters show the great potential of rank metric for cryptography
especially for higher security settings.
When compared to the ring-LPN based cryptosystem \cite{damgaard2012public}
our system has better parameters with factors $10$ and $100$ respectively for the size of the ciphertext and the size of the public key. 
We also give a general table comparing the different asymptotic sizes for different code-based cryptosystems.

\subsection{Overview of Our Techniques}
\label{sec:1-3-overview}
Our cryptosystem is based on two codes. A first code $\mathcal{C}[n,k]$, for which an efficient decoding algorithm 
$\mathcal{C}.\textsf{Decode}(\cdot)$ is known.
The code $\mathcal{C}$ together with its generator matrix $\mathbf{G}$ 
are publicly known.
The second code is a $[2n,n]$ random double-circulant code in systematic form, with generator matrix $\mathbf{Q} = (\mathbf{I}_n~|~\textbf{rot}(\mathbf{q}_r))$ (see Eq. (\ref{eq:rot}) for the definition of \textbf{rot}($\cdot$)). The general idea of the system is that the double-circulant code is used to generate some noise, which can be handled 
and decoded by the code $\mathcal{C}$. The system can be seen as a noisy adaptation of the ElGamal cryptosystem.

The secret key for our cryptosystem is a \emph{short}
vector $\mathsf{sk} = \mathbf{(x, y)}$ (for some metric), whose syndrome 
$\mathbf{s}^\top = \mathbf{Q(x, y)}^\top$ is appended to the public key 
$\mathsf{pk} = (\mathbf{G, Q, s}^\top)$. To encrypt a message $\bm{\mu}$ belonging 
to some plaintext space, it is first encoded through the generator matrix $\mathbf{G}$, 
then hidden using the syndrome $\mathbf{s}$ and an additional short vector $\bm{\epsilon}$ 
to prevent information leakage. In other words, encrypting a message simply consists in 
providing a noisy encoding of it with a particular shape. Formally, the ciphertext is $(\mathbf{v=rQ}^\top,  \bm{\rho})$, for a short random vector $\mathbf{r = (r}_1, \mathbf{r}_2)$ and  
$\bm{\rho} = \bm{\mu}\mathbf{G} + \mathbf{s\cdot r}_2 + \bm{\epsilon}$ for some natural operator $\cdot$ 
defined in Sec.~\ref{sec:preliminaries}. 
The legitimate recipient can obtain a noisy version of the plaintext $\bm{\rho} - \mathbf{v\cdot y}$ using his secret key $\textsf{sk} = \mathbf{(x,y)}$ and then recover the (noiseless) plaintext using the
efficient decoding algorithm $\mathcal{C}.\textsf{Decode}$.

For correctness, all previous constructions based on a McEliece
approach rely on the fact that the error term 
added to the encoding of the message is less than or equal to the decoding capability 
of the code being used. In our construction, this assumption is no longer required 
and the correctness of our cryptosystem is guaranteed assuming the legitimate recipient 
can remove sufficiently many errors from the noisy encoding $\bm{\rho}$ of the message 
using \textsf{sk}. 

The above discussion leads to the study of the probability that a decoding error occurs, 
which would yield a decryption failure. We study the typical weight of the
error vector $\e$ that one needs to decode in order to decrypt 
(see Sec.~\ref{sec:5-analysis} for details). With
the reasonable assumption, backed up by simulations, that the weight
of $\e$ behaves in a way that is close to a binomial distribution, we
manage a precise estimation of a decoding failure and hence calibrate coding
parameters accordingly.

\smallskip\noindent{\bf Comparison with the McEliece framework. } In the McEliece encryption framework, a hidden code is considered.
This leads to two important consequences: first, the security depends on hiding the structure of the code, and second,
the decryption algorithm consists of decoding the hidden code which
cannot be changed. This yields different instantiations 
depending on the choice of the hidden code, many of which succumb to
attacks and few of which resist.

In our framework there is not one unique hidden code, but two independent codes:  
the random double-circulant structure guarantees the security of the scheme, and the public code
$\mathcal{C}$ guarantees correct decryption. It makes it possible to
consider public families of codes 
which are difficult to hide but very efficient for decoding: also it requires        
finding a tradeoff for the code $\mathcal{C}$, between decoding
efficiency and practical decoding complexity.
But unlike the McEliece scheme, where the decryption code is fixed, it can be changed
depending on the application.

The global decryption failure for our scheme depends on the articulation between the error-vector distribution induced
by the double-circulant code
and the decoding algorithm $\mathcal{C}.\textsf{Decode}(\cdot)$. After having studied the error-vector distribution
for the Hamming metric we associate it with a particular code adapted
to low rates  and bit error probability of order $1/3$. 
Notice that the system could possibly be used for greater encryption rate at the cost of higher parameters.
This led us to choose tensor product codes, the composition of two linear codes. 
Tensor product codes are defined (Def.~\ref{def:tensor}) in Sec.~\ref{sec:6-low-rate}, and 
a detailed analysis of the decryption failure probability for such codes is provided there.
For the rank metric case, we consider Gabidulin codes and the case when the error-vector is always decodable,
with zero decryption failure probability.

\smallskip\noindent{\bf Comparison with the LWE/LPN approach.}
Our scheme may be considered as a special instance
of the general LWE/LPN
methodology, as described, for example, in the recent paper
\cite {bSBDIRZ16}. As is mentioned there, even though full LWE-based
schemes may, given current knowledge, be asymptotically more efficient
than their LPN counterparts, there is still significant appeal in
providing a workable variation over the more simple binary field
(as it was done with Ring-LWE for the LWE setting).
This was previously attempted in \cite{damgaard2012public} by relying on the
Ring-LPN problem. One of the drawbacks of this last work is to be
limited to rings
of the form $\f_2[X]/(P(X))$ that are extension fields of $\f_2$. In
contrast, we suggest using $\f_q[X]/(X^n-1)$, which reduces security
to a decoding problem for quasi-cyclic codes and draws upon Coding
Theory's experience of using this family of codes.
Quasi-cyclic codes
have indeed been studied for a long time by coding-theorists, and many of the
records for minimum distance are held by quasi-cyclic codes. However,
no efficient generic decoding algorithm for quasi-cyclic codes has
been found, lending faith to the assumption that decoding random
quasi-cyclic codes is a hard algorithmic problem.
Also, this particular setting
also allows us to obtain very good parameters compared to the approach of \cite{damgaard2012public} with at least
a factor 10 for the size of the keys and messages
Departing from the strict
LWE/LPN paradigm also enabled us to derive a security reduction to decoding
quasi-cyclic codes and arguably gives us more flexibility for the
error model. Notably the rank-metric variation that we introduce has
not been investigated before in the LWE/LPN setting, and looks very
promising. As mentioned before, one of its features is that it enables
a zero error probability of incorrect decryption.

\subsection{Road Map}

The rest of the paper is organized as follows: Sec.~\ref{sec:preliminaries} gives necessary background on coding theory
for Hamming and Rank metrics. Sec.~\ref{sec:scheme} describes the cryptosystem we propose and its security is discussed in Sec.~\ref{sec:secu}. Sec.~\ref{sec:5-analysis} and~\ref{sec:6-low-rate} study
the decryption failure probability and the family of tensor product codes we consider to perform
the decoding for small rate codes. Finally, Sec.~\ref{sec:7-parameters} give parameters.

%\begin{itemize}
%\item High level overview of the scheme
%\item Tensor product code
%\item Decryption Failure Analysis overview ?
%\item Overview of key exchange
%\end{itemize}
\section{Preliminaries}
\label{sec:preliminaries}

\subsection{General Definitions}
\noindent{\bf Notation. } Throughout this paper, $\mathbb{Z}$ denotes the ring of integers, $\mathbb{F}$ denotes a finite (hence commutative) field, typically $\mathbb{F}_q$ for a prime $q \in \mathbb{Z}$ for Hamming codes or $\mathbb{F}_{q^m}$ for Rank Metric codes. $\mathcal{V}$ is a vector space of dimension $n$ over $\mathbb{F}$ for some positive $n \in \mathbb{Z}$. Elements of $\mathcal{V}$ will be represented by lower-case bold letters, and interchangeably considered as row vectors or polynomials in $\mathcal{R}=\mathbb{F}[X]/(X^{n}-1)$. By extension $\mathcal{R}_q$ and $\mathcal{R}_{q^m}$ will denote the latter ring when the base field is $\mathbb{F}_q$ or $\mathbb{F}_{q^m}$ instead of $\mathbb{F}$, respectively. Matrices will be represented by upper-case bold letters. 

For any two elements $\mathbf{x, y} \in \mathcal{V}$, we define their product similarly as in $\mathcal{R}$, \emph{i.e.} $\mathbf{x}\cdot\mathbf{y} = \mathbf{c} \in \mathcal{V}$ with 
\begin{equation}
\label{eq:product}
c_k = \sum_{i+j\equiv k \mod~n} x_i y_j\textnormal{, for }k \in \{0,1, \ldots, n-1\}.
\end{equation}
Notice that as the product of two elements over the \emph{commutative} ring $\mathcal{R}$, we have $\mathbf{x}\cdot\mathbf{y} = \mathbf{y}\cdot\mathbf{x}$.

For any finite set $\mathcal{S}$, $x \lar \mathcal{S}$ denotes a uniformly random element sampled from $\mathcal{S}$. For any $x \in \mathbb{R}$, let $\lfloor x \rfloor$ denotes the biggest integer smaller than (or equal to) $x$. Finally, all logarithms $\log(\cdot)$ will be base-$2$ unless explicitly mentioned. For a probability distribution $\mathcal{D}$, we denote by $X \sim \mathcal{D}$ the fact that $X$ is a random variable following $\mathcal{D}$.

\begin{definition}[Circulant Matrix]
Let $\mathbf{x} = \left(x_1, \ldots, x_n\right) \in \mathbb{F}^n$. The \emph{circulant matrix} induced by $\mathbf{x}$ is defined and denoted as follows:
\begin{equation}\label{eq:rot}
\textnormal{\textbf{rot}}(\mathbf{x}) = \begin{pmatrix}
x_1 & x_n & \ldots & x_2\\
x_2 & x_1 & \ldots & x_3\\
\vdots & \vdots & \ddots & \vdots\\
x_n & x_{n-1} & \ldots & x_1\\
\end{pmatrix} \in \mathbb{F}^{n \times n}
\end{equation}
\end{definition}

As a consequence, it is easy to see that the product of any two elements $\mathbf{x, y} \in \mathcal{V}$ can be expressed as a usual vector-matrix (or matrix-vector) product using the \textbf{rot}$(\cdot)$ operator as 
\begin{equation}
\mathbf{x}\cdot\mathbf{y} = \mathbf{x}.\textnormal{\textbf{rot}}(\mathbf{y})^\top = \left(\textnormal{\textbf{rot}}(\mathbf{x})\mathbf{y}^\top\right)^\top = \mathbf{y}.\textnormal{\textbf{rot}}(\mathbf{x})^\top = \mathbf{y}\cdot\mathbf{x}.
\end{equation}

\bigskip\noindent{\bf Coding Theory. } We now turn to recall some basic definitions and properties relating to coding theory that will be useful to our construction. We mainly focus on generic definitions, and refer the reader to Sec.~\ref{sec:2-2-metrics} for instantiations with a specific metric, and also, to~\cite{background} for a complete survey on Code-based Cryptography due to space restrictions.

\philippe{Objectif $\rightarrow$ enlever des défs, et les remplacer par des paragraphes.}
\begin{definition}[Linear Code]
A \emph{Linear Code} $\mathcal{C}$ of length $n$ and dimension $k$ (denoted $[n,k]$) is a subspace of $\mathcal{V}$ of dimension $k$. 
Elements of $\mathcal{C}$ are referred to as codewords.
\end{definition}

\begin{definition}[Generator Matrix]
\label{def:GM}
We say that $\mathbf{G} \in \mathbb{F}^{k \times n}$ is a \emph{Generator Matrix} for the $[n, k]$ code $\mathcal{C}$ if \begin{equation}
\mathcal{C} = \left\lbrace \bm{\mu}\mathbf{G}\textnormal{, for }\bm{\mu} \in \mathbb{F}^{k}\right\rbrace.
\end{equation}
\end{definition}

\begin{definition}[Parity-Check Matrix]
\label{def:PCM}
Given an $[n, k]$ code $\mathcal{C}$, we say that $\mathbf{H} \in \mathbb{F}^{(n-k) \times n}$ is a \emph{Parity-Check Matrix} for $\mathcal{C}$ if $\mathbf{H}$ is a generator matrix of the dual code $\mathcal{C}^\perp$, or more formally, if
\begin{equation}
\mathcal{C}^\perp = \left\lbrace \mathbf{x}\in \mathbb{F}^{n}\textnormal{ such that }\sigma(\mathbf{x}) = \mathbf{0}\right\rbrace.
\end{equation}
where
$$\sigma(\mathbf{x}) = \mathbf{Hx}^\top$$
denotes the {\em syndrome} of $\mathbf{x}$.
\end{definition}
\begin{definition}[Minimum Distance]
Let $\mathcal{C}$ be an $[n,k]$ linear code over $\mathcal{V}$ and let $\omega$ be a norm on $\mathcal{V}$. The \emph{Minimal Distance} of $\mathcal{C}$ is
\begin{equation}
d = \min_{\mathbf{x, y} \in \mathcal{C}, \mathbf{x \neq y}} \omega(\mathbf{x-y}).
\end{equation}
\end{definition}

A code with minimum distance $d$ is capable of decoding arbitrary patterns of up to 
$\delta = \lfloor \frac{d-1}{2} \rfloor$ errors. Code parameters are written
denoted $[n, k, d]$.

Code-based cryptography usually suffers from huge keys. In order to keep our cryptosystem efficient, we will use the strategy of Gaborit~\cite{cyclic} for shortening keys. This results in Quasi-Cyclic Codes, as defined below.

\begin{definition}[Quasi-Cyclic Codes~\cite{MTSB13}]
\label{def:QC}
View a vector $\mathbf{x}=(\mathbf{x}_1,\ldots ,\mathbf{x}_s)$ of $\f_2^{sn}$ as 
$s$ successive blocks ($n$-tuples). 
An $[sn, k, d]$ linear code $\mathcal{C}$ is \emph{Quasi-Cyclic (QC)} of order $s$ if, for any $\mathbf{c}=(\mathbf{c}_1,\ldots ,\mathbf{c}_s) \in \mathcal{C}$, the vector obtained after applying a simultaneous circular shift to every block 
$\mathbf{c}_1,\ldots ,\mathbf{c}_s$
 is also a codeword. 

More formally, by considering each block $\mathbf{c}_i$ as a
polynomial in $\mathcal{R} = \mathbb{F}[X]/(X^n-1)$, the code
$\mathcal{C}$ is QC of order $s$ if for any
$\mathbf{c}=(\mathbf{c}_1,\ldots ,\mathbf{c}_s) \in \mathcal{C}$ it
holds that $(X\cdot\mathbf{c}_1,\ldots , X\cdot\mathbf{c}_s) \in
\mathcal{C}$.
\end{definition}

\begin{definition}[Systematic Quasi-Cyclic Codes]
\label{def:QCsystematic}
A {\em systematic} Quasi-Cyclic $[sn,(s-\ell)n]$ code of order $s$ is a quasi-cyclic code
with a parity-check matrix of the form:
\begin{equation}
  \label{eq:systematic}
  \mathbf{H}=
\begin{bmatrix}
  \mathbf{I}_n & 0   &  \cdots &0& \mathbf{A}_1\\
  0   & \mathbf{I}_n &        && \mathbf{A}_2\\
      &     &  \ddots && \vdots\\
  0   &     &   \cdots & \mathbf{I}_n     & \mathbf{A}_{\ell}
\end{bmatrix}
\end{equation}
where $\mathbf{A}_1,\ldots ,\mathbf{A}_\ell$ are circulant $n\times n$ matrices.
\end{definition}

\subsection{Different Types of Metric}
\label{sec:2-2-metrics}
The previous definitions are generic and can be adapted to any type of metric.

Besides the well known Hamming metric, 
we also consider, in this paper, the rank metric which has interesting properties for cryptography.

We recall some definitions and properties of Rank Metric Codes, and refer the reader to~\cite{L06} for more details. Consider the case where $\mathbb{F}$ is an extension of a finite field, \emph{i.e.} $\mathbb{F = F}_{q^m}$, and let $\mathbf{x} = (x_1, \ldots, x_n) \in \mathbb{F}_{q^m}^n$ be an element of some vector space $\mathcal{V}$ of dimension $n$ over $\mathbb{F}_{q^m}$. A basic property of field extensions is that they can be seen as vector spaces over the base field they extend. Hence, by considering $\mathbb{F}_{q^m}$ as a vector space of dimension $m$ over $\mathbb{F}_q$, and given a basis $\left(\mathbf{e_1, \ldots, e_m}\right)$ $\in \mathbb{F}_q^m$, one can express each $x_i$ as
\begin{eqnarray}
x_i = \sum_{j=1}^m x_{j,i}\mathbf{e_j} \textnormal{ (or equivalently } x_i = \left(x_{1,i}, \ldots, x_{m, i}\right)\textnormal{).}
\end{eqnarray}
Using such an expression, we can expand $\mathbf{x} \in \mathbb{F}_{q^m}^n$ to a matrix $\mathbf{E(x)}$ such that:
\begin{eqnarray}
\mathbf{x} &=& ~\begin{pmatrix}~x_1~~ & ~x_2 & ~\ldots & ~x_n~~\end{pmatrix} \in \mathbb{F}_{q^m}^n\\
\label{eq:rank_matrix}\mathbf{E(x)}&=& \begin{pmatrix}
x_{1, 1} & x_{1, 2} & \ldots & x_{1, n} \\
x_{2, 1} & x_{2, 2} & \ldots & x_{2, n} \\
\vdots & \vdots & \ddots & \vdots \\
x_{m, 1} & x_{m, 2} & \ldots & x_{m, n}
\end{pmatrix} \in \mathbb{F}_q^{m \times n}.
\end{eqnarray}

The definitions usually associated to Hamming metric codes such as norm (Hamming weight), support (non-zero coordinates), and isometries ($n\times n$ permutation matrices) can be adapted to the Rank metric setting based on the representation of elements as matrices in $F_{q}^{m\times n}$.

For an element $\mathbf{x}$ of $\mathbb{F}_{q^m}^n$ we define its rank norm $\omega(\mathbf{x})$ as the rank of the matrix $\mathbf{E(x)}$. A rank metric code $\mathcal{C}$ of length $n$ and dimension $k$ over the field $\mathbb{F}_{q^m}$ is a subspace of dimension $k$ of $\mathbb{F}_{q^m}^n$ embedded with the rank norm. In the following, $\mathcal{C}$ is a rank metric code of length $n$ and dimension $k$ over
$\mathbb{F}_{q^m}$, where $q = p^\eta$ for some prime $p$ and positive $\eta \geq 1$. The matrix $\mathbf{G}$ denotes a $k \times n$ generator matrix of $\mathcal{C}$ and $\mathbf{H}$ is one of its parity check matrices.
The minimum rank distance of the code $\mathcal{C}$ is the minimum rank of non-zero
vectors of the code. We also considers the usual inner product which allows
to define the notion of dual code.

Let $\mathbf{x}=\left(x_1,x_2,\cdots,x_n\right) \in \mathbb{F}_{q^m}^n$ be a vector of rank $r$. We denote by $E = \langle x_1,\ldots, x_n \rangle$ the
$\mathbb{F}_{q}$-subspace of $\mathbb{F}_{q^m}$ generated by the coordinates of $\mathbf{x}$ \emph{i.e.} $E = \mathrm{Vect}\left(x_1,\ldots,x_n \right)$. The vector space $E$ is called the \emph{support} of $\mathbf{x}$ and denoted $\mathrm{Supp}(\mathbf{x})$.
Finally, the notion of \emph{isometry} which in Hamming metric corresponds to the action of the code on $n \times n$ permutation matrices, is replaced for the Rank metric by the action of $n \times n$ invertible matrices
over the base field $\mathbb{F}_q$.

\smallskip\noindent{\bf Bounds for Rank Metric Codes. } The classical bounds for Hamming metric have straightforward rank metric analogues.

\smallskip\noindent{\bf Singleton Bound. } The classical Singleton bound for linear $[n,k]$ codes of minimum rank
$r$ over $\mathbb{F}_{q^m}$ applies naturally in the Rank metric setting. It works in the same way as for linear codes (by finding an information set) and reads $r \le 1+n-k$.
When $n > m$ this bound can be rewritten~\cite{L06} as
\begin{equation}
r \le 1+ \left\lfloor \frac{(n-k)m}{n} \right\rfloor.
\end{equation}
Codes achieving this bound are called Maximum Rank Distance codes (MRD).

\medskip\noindent{\bf Deterministic Decoding. } Unlike the situation for the Hamming metric, there do not exist many families of codes for the rank metric which are able to decode rank errors efficiently up to a given norm. When we are dealing with deterministic decoding,
there is essentially only one known family of rank codes which can decode efficiently: the family of Gabidulin codes~\cite{G85}. These codes are an analogue of Reed-Solomon
codes~\cite{RS60} where polynomials are replaced by $q$-polynomials. These codes are defined over $\mathbb{F}_{q^m}$
and for $k\le n \le m$, Gabidulin codes of length $n$ and dimension $k$ are optimal and satisfy the
Singleton bound for $m=n$ with minimum distance $d=n-k+1$. They can decode
up to $\lfloor\frac{n-k}{2}\rfloor$ rank errors in a deterministic way.

\medskip\noindent{\bf Probabilistic Decoding. } There also exists a simple family of codes
which has been described for the subspace metric in~\cite{SKK10} and can be straightforwardly adapted
to rank metric. These codes reach asymptotically the equivalent of the Gilbert-Varshamov bound for the rank metric, however their non-zero probability of decoding failure makes them less interesting
for the cases we consider in this paper.

\subsection{Difficult Problems for Cryptography}

In this section we describe difficult problems which can be used for cryptography.
We give generic definitions for these problems which are usually instantiated with the Hamming metric
but can also be instantiated with the rank metric. After defining the problems
we discuss their complexity.

All problems are variants of the {\em decoding problem,} which consists of
looking for the closest codeword to a given vector: when dealing with linear codes, it is readily seen that the decoding problem stays the same when one is given the {\em syndrome} of the received vector rather than the received vector. We therefore speak of {\em Syndrome Decoding} (SD).

\begin{definition}[SD Distribution]
For positive integers, $n$, $k$, and $w$, the \emph{\textsf{SD}($n, k, w$) Distribution} chooses $\mathbf{H} \lar \mathbb{F}^{(n-k) \times n}$ and $\mathbf{x} \lar \mathbb{F}^{n}$ such that $\omega(\mathbf{x})=w$, and outputs $(\mathbf{H, \sigma(\x)=Hx}^\top)$.
\end{definition}

\begin{definition}[Search SD Problem]
\label{def:SD}
Let $\omega$ be a norm over $\mathcal{V}$. On input $(\mathbf{H, y}^\top) \in \mathbb{F}^{(n-k)\times n}\times\mathbb{F}^{(n-k)}$ from the \SD distribution, the \emph{Syndrome Decoding Problem \textsf{SD}$(n, k, w)$} asks to find $\mathbf{x} \in \mathbb{F}^n$ such that $\mathbf{Hx}^\top = \mathbf{y}^\top$ and $\omega(\mathbf{x}) = w$.
\end{definition}

Depending on the metric the above problem is instantiated with, we denote it either by \SD for the Hamming metric or by Rank-\textsf{SD} (\RSD) for the Rank metric.

For the Hamming distance the \SD problem has been proven to be \textsf{NP}-complete in~\cite{BMvT78}. This problem can also be seen as the Learning Parity with Noise (\textsf{LPN}) problem with a fixed number of samples~\cite{C:AppIshKus07}. The \RSD problem has recently
been proven difficult with a probabilistic reduction to the Hamming setting in~\cite{GZ14}. For cryptography 
we also need a Decisional version of the problem, which is given in the following Definition:

\begin{definition}[Decisional SD Problem]
\label{def:DSD}
On input $(\mathbf{H, y}^\top) \lar \mathbb{F}^{(n-k)\times n}\times\mathbb{F}^{(n-k)}$, the \emph{Decisional \SD Problem \textsf{DSD}($n, k, w$)} asks to decide with non-negligible advantage whether $(\mathbf{H, y}^\top)$ came from the \textnormal{\textsf{SD}}($n, k, w$) distribution or the uniform distribution over $\mathbb{F}^{(n-k)\times n}\times\mathbb{F}^{(n-k)}$.
\end{definition}

As mentioned above, this problem is the problem of decoding random linear codes from random errors. The random errors are often taken as independent
Bernoulli variables acting independently on vector coordinates, rather
than uniformly chosen from the set of errors of a given weight, but this
hardly makes any difference and one model rather than the other is a question of convenience.
The \DSD 
problem has been shown to be polynomially equivalent to its
search version in~\cite{C:AppIshKus07}. The rank metric version of the problem is denoted by \DRSD, by applying the transformation
described in~\cite{GZ14} it can be shown that the problem can be reduced to a search problem for the Hamming metric. Hence even if the reduction
is not optimal, it nevertheless shows the hardness of the problem.

%\begin{assumption}\label{as:DSDhard}
%For certain choices of $n(\lambda)$, $k(\lambda)$, and $w(\lambda)$, the \textsf{DSD}$(n, k, w)$ problem is intractable in the average-case. \jc{required ? \at{Philippe} au final je crois pas qu'on l'utilise celle-ci, si ?}
%\end{assumption}

%This assumption is widely believed to hold~\cite{Ste96,Meu13} \jc{...}

Finally, as for both metrics our cryptosystem will use QC-codes, we explicitly define the problem on which our cryptosystem will rely. The following Definitions describe the \DSD problem in the QC configuration, and are just a combination of Def.~\ref{def:QC} and~\ref{def:DSD}.
Quasi-Cyclic codes are very useful in cryptography since their compact description allows to decrease considerably the size of the keys. In particular the case
$s=2$ corresponds to double circulant codes with generator
matrices of the form $(\mathbf{I}_n~|~\mathbf{A})$ for $\mathbf{A}$ a circulant matrix. Such double circulant
codes have been used for almost 10 years in cryptography (cf~\cite{GG07}) and more
recently in~\cite{MTSB13}. Quasi-cyclic codes of order 3 are also considered in~\cite{MTSB13}.

\begin{definition}[$s$-QCSD Distribution]
\label{def:QCSD-distro}
For positive integers $n$, $k$, $w$ and $s$, the \emph{$s$-\textsf{QCSD}($n, k, w, s$) Distribution} chooses uniformly at random a parity matrix $\mathbf{H} \lar \mathbb{F}^{(sn-k) \times sn}$ of a systematic QC code
$\mathcal{C}$ of order $s$ (see Definition~\ref{def:QCsystematic})  together with a vector $\mathbf{x}=(\mathbf{x}_1,\ldots , \mathbf{x}_s) \lar \mathbb{F}^{sn}$ such that $\omega(\mathbf{x}_i)=w$, $i=1..s$, and outputs $(\mathbf{H, Hx}^\top)$.
\end{definition}
\begin{definition}[(Search) $s$-QCSD Problem]
\label{def:SQCSD}
For positive integers $n$, $k$, $w$, $s$, a random parity check matrix $\mathbf{H}$ of a systematic QC code $\mathcal{C}$ and $\mathbf{y} \lar \mathbb{F}^{sn-k}$, the \emph{Search $s$-Quasi-Cyclic \SD Problem $s$-\textsf{QCSD}($n, k, w$)} asks to find 
$\mathbf{x}=(\mathbf{x}_1,\ldots , \mathbf{x}_s) \in \mathbb{F}^{sn}$ such that $\omega(\mathbf{x}_i) =w$, $i=1..s$, and $\mathbf{y} = \mathbf{xH}^\top$.
\end{definition}

It would be somewhat more natural to choose the parity-check matrix $\mathbf{H}$ to be made up of independent uniformly random circulant submatrices, rather than with the special form required by \eqref{eq:systematic}. We choose this distribution so as to make the
security reduction to follow less technical. It is readily seen that, for fixed $s$, 
when choosing quasi-cyclic codes with this more general distribution,
one obtains with non-negligeable probability, 
a quasi-cyclic code that admits a parity-check matrix of the form \eqref{eq:systematic}. Therefore requiring quasi-cyclic codes to be systematic does not  hurt the generality of the decoding problem for quasi-cyclic codes.
A similar remark holds for the slightly special form of weight distribution of the vector $\mathbf{x}$.

\begin{assumption}
Although there is no general complexity result for quasi-cyclic codes,
decoding these codes is considered hard by the community.
There exist general attacks which uses the cyclic structure of the code
\cite{S11,HT15} but these attacks have only a very limited impact on
the practical complexity of the problem. The conclusion is that in practice, the best
attacks are the same as those for non-circulant codes up to a small
factor. \philippe{papier JPT.}
\end{assumption}

The problem has a decisional form:
\begin{definition}[Decisional $s$-QCSD Problem]
\label{def:DQCSD}
For positive integers $n$, $k$, $w$, $s$, a random parity check matrix $\mathbf{H}$ of a systematic QC code $\mathcal{C}$ and $\mathbf{y} \lar \mathbb{F}^{sn}$, the \emph{Decisional $s$-Quasi-Cyclic \SD Problem $s$-\textsf{DQCSD}($n, k, w$)} asks to decide with non-negligible advantage whether $(\mathbf{H, y}^\top)$ came from the $s$-\textsf{QCSD}($n, k, w$) distribution or the uniform distribution over $\mathbb{F}^{(sn-k)\times sn}\times\mathbb{F}^{sn-k}$.
\end{definition}

As for the ring-LPN problem, there is no known reduction from the search version of $s$-QCSD problem to its decisional version.
The proof of \cite{C:AppIshKus07} cannot be directly adapted in the quasi-cyclic case, however the best known attacks on the decisional
version of the problem $s$-QCSD remain the direct attacks on the search version of the problem $s$-QCSD. 

The situation is similar for the rank versions of these problems which are respectively denoted by $s$-\RQCSD and $s$-\DRQCSD,
and for which the best attacks over the decisional problem consist in
attacking the search version of the problem.

\subsection{Practical Attacks}
\label{sec:2-4-attacks}

The practical complexity of the \SD problem for the Hamming metric has been widely studied
for more than 50 years. For small weights the best known attacks are
exponential in the weight of the researched codeword.
The best attacks can be found in~\cite{EC:BJMM12}.

The \RSD problem is less known in cryptography but has also been
studied for a long time, ever since a rank metric version of the McEliece 
cryptosystem was introduced in 
1991~\cite{EC:GabParTre91}.
We recall the main types of attack on the \RSD problem below.

The complexity of practical attacks grows very quickly with the size of parameters:
there is a structural reason to this. For the Hamming distance,
attacks typically rely on enumerating
the number of words of length $n$ and support size (weight) $t$,
which amounts to the Newton binomial coefficient $\binom{n}{t}$, 
whose value is bounded from above by
by $2^n$. In the rank metric case, counting the number of possible supports of size
$r$ for a rank code of length $n$ over $\mathbb{F}_{q^m}$ corresponds to counting the number
of subspaces of dimension $r$ in $\mathbb{F}_{q^m}$: this involves
the \emph{Gaussian binomial coefficient}
of size roughly $q^{(m-r)m}$, whose value is also exponential 
in the blocklength but with a quadratic term
in the exponent.

There exist two types of generic attacks on the problem:
\begin{itemize}
\item \textbf{Combinatorial attacks}: these attacks are usually the best ones for small values
of $q$ (typically $q=2$) and when $n$ and $k$ are not too small:
when $q$ increases, the combinatorial aspect makes them less efficient.
The best combinatorial attack has recently been updated to $(n-k)^3m^3q^{(r-1)\lfloor\frac{(k+1)m}{n}\rfloor}$ to
take into account the value of $n$~\cite{GRS16}.
\item \textbf{Algebraic attacks}: the particular nature of the rank metric makes it a natural field
for algebraic attacks using Gröbner bases, since these attacks 
are largely independent of the value of $q$
and in some cases may also be largely independent of $m$.
These attacks are usually the most efficient when $q$ increases.
For the cases considered in this paper where $q$ is taken to be small,
the complexity is greater than the cost of combinatorial attacks (see~\cite{LP06,C:FauLevPer08,GRS16}).
\end{itemize}

Note that the recent improvements on decoding random codes for the
Hamming distance correspond to birthday paradox attacks. An open
question is whether these improvements apply to rank metric
codes. Given that the support of the error on codewords in rank metric
is not related to the error coordinates, the birthday paradox strategy
has failed for the rank metric, which for the moment seems to keep these codes protected from the aforementioned  advances.

\section{A New Encryption Scheme}
\label{sec:scheme}
\subsection{Encryption and Security}
%\jc{\at{Olivier} C'était pas formel les définitions au dessus ? Si tu veux qu'on introduise d'autres choses formelles qui ne concernent ni le background maths ni la partie code on peut peut-être mettre une \textbackslash\texttt{subsection}\{Cryptosystem and Security\}, non ?}

%\ob{Mettre key exchange, ???}
\noindent{\bf Encryption Scheme. } An encryption scheme is a tuple of four polynomial time algorithms $(\Setup, \KeyGen,\Encrypt, \Decrypt)$:
\begin{itemize}
	\item $\Setup(1^\seck)$, where $\seck$ is the security parameter, generates the global parameters $\param$ of the scheme;
	\item $\KeyGen(\param)$ outputs a pair of keys, a (public) encryption key  $\pk$ and a (private) decryption key $\sk$;
	\item $\Encrypt(\pk, \bm{\mu}, \theta)$ outputs a ciphertext $\bf{c}$, on the message $\bm{\mu}$, under the encryption key $\pk$, with the randomness $\theta$;
	\item $\Decrypt(\sk,\bf{c})$ outputs the plaintext $\bm{\mu}$, encrypted in the ciphertext $\bf{c}$ or $\bot$.
\end{itemize}
Such an encryption scheme has to satisfy both \emph{Correctness} and \emph{Indistinguishability under Chosen Plaintext Attack} (IND-CPA) security properties.

\medskip\noindent{\bf Correctness}: For every $\lambda$, every
$\param\leftarrow \Setup(1^\lambda),$ every pair of keys $(\pk,\sk)$
generated by $\KeyGen$, every message $\bm{\mu}$, we should have $P[\Decrypt(\sk,\Encrypt(\pk,\bm{\mu}, \theta)) = \bm{\mu}]= 1 - \epsilon(\lambda)$ for $\epsilon$ a negligible function, where the probability is taken over varying randomness $\theta$.

\medskip\noindent{}\begin{minipage}[ht]{.56\textwidth}
		{\bf IND-CPA}~\cite{GolMic84}: This notion formalized by the adjacent game, states that an adversary shouldn't be able to efficiently guess which plaintext has been encrypted even if he knows it is one among two plaintexts of his choice.
		
		The global advantage for polynomial time adversaries (running in time less than $t$) is:
\begin{equation}
			\Adv^{\ind}_{\Ee}(\seck,t) = \max_{\Ae \leq t} \Adv^{\ind}_{\Ee,\Ae}(\seck), 
\end{equation}

		  \end{minipage}\hfill
	\framebox{\begin{minipage}[ht]{.40\textwidth}
			\begin{tabbing}
				= \= ==== \= = \= \kill \\[-1em]
				$\Exp_{\Ee,\Ae}^{\ind-b}(\seck)$ \\
				1. \> $\param \sets \Setup(1^\seck)$ \\
				2. \> $(\pk,\sk) \sets \KeyGen(\param)$ \\
				3. \> $(\bm{\mu}_0, \bm{\mu}_1) \sets \Ae(\FIND: \pk)$ \\
				4. \> $\bf{c}^* \sets \Encrypt(\pk,\bm{\mu}_b,\theta)$ \\
				5. \> $b' \sets \Ae(\GUESS:\bf{c}^*)$ \\
				6. \> \comreturn $b'$
			\end{tabbing}
		\end{minipage}}

\medskip

\noindent
where $\Adv^{\ind}_{\Ee,\Ae}(\seck)$ is the advantage the adversary $\mathcal{A}$ has in winning game $\Exp_{\Ee,\Ae}^{\ind-b}(\seck)$:
\begin{equation}
\Adv^{\ind}_{\Ee,\Ae}(\seck) = \left| \Pr[\Exp_{\Ee,\Ae}^{\ind-1}(\seck) = 1] - \Pr[\Exp_{\Ee,\Ae}^{\ind-0}(\seck) = 1] \right|.
\end{equation}

\subsection{Presentation of the Scheme}
\label{sec:3-2-presentation}
We begin this Section by describing a generic version of the proposed encryption scheme. This description does not depend
on the particular metric used. The particular case of the Hamming metric is denoted by HQC (for Hamming Quasi-Cyclic) and RQC (for Rank Quasi-Cyclic) in the case of the rank metric. Parameter sets for binary Hamming Codes and Rank Metric Codes can be respectively found in Sec.~\ref{sec:7-1-hamming} and~\ref{sec:7-2-rank}.

\medskip\noindent{\bf Presentation of the scheme. } Recall from the introduction that the scheme uses two types of codes, 
a decodable $[n,k]$ code which can correct $\delta$ errors 
and a random double-circulant $[2n,n]$ code. 
In the following, we assume $\mathcal{V}$ is a vector space on some field $\mathbb{F}$, $\omega$ is a norm on $\mathcal{V}$ and for any $\mathbf{x}$ and $\mathbf{y} \in \mathcal{V}$, their distance is defined as $\omega(\mathbf{x-y}) \in \mathbb{R}^+$. Now consider a linear code $\mathcal{C}$ over $\mathbb{F}$ of dimension $k$ and length $n$ (generated by $\mathbf{G} \in \mathbb{F}^{k\times n}$), that can correct up to $\delta$ errors via an efficient algorithm $\mathcal{C}$.\textsf{Decode}$(\cdot)$. The scheme consists of the following four polynomial-time algorithms:
\begin{itemize}
\item \Setup$(1^\lambda)$: generates the global parameters $n = n(1^\lambda), k = k(1^\lambda)$, $\delta = \delta(1^\lambda)$, and $w = w(1^\lambda)$. The plaintext space is $\mathbb{F}^{k}$. Outputs \param{} = $(n, k, \delta, w)$.
\item \KeyGen$(\param)$: generates $\mathbf{q}_r \lar \mathcal{V}$, matrix $\mathbf{Q} = \left(\mathbf{I}_n~|~\textnormal{\textbf{rot}}(\mathbf{q}_r)\right)$, the generator matrix $\mathbf{G}\in \mathbb{F}^{k\times n}$ of $\mathcal{C}$, $\sk = (\mathbf{x}, \mathbf{y}) \lar \mathcal{V}^2$ such that $\omega(\mathbf{x})= \omega(\mathbf{y}) = w$, sets $\pk = \left(\mathbf{G, Q}, \mathbf{s} = \sk\cdot\mathbf{Q}^\top\right)$, and returns $(\pk, \sk)$.
\item \Encrypt$(\pk=(\mathbf{G, Q}, \mathbf{s}), \bm{\mu}, \theta)$: uses randomness $\theta$ to generate $\bm{\epsilon} \lar \mathcal{V}$, $\mathbf{r} = (\mathbf{r}_1, \mathbf{r}_2) \lar \mathcal{V}^2$ such that $\omega(\bm{\epsilon}), \omega(\mathbf{r}_1), \omega(\mathbf{r}_2) \leq w$, sets $\mathbf{v}^\top = \mathbf{Qr}^\top$ and $\bm{\rho} = \bm{\mu}\mathbf{G} + \mathbf{s\cdot r}_2 + \bm{\epsilon}$. It finally returns $\mathbf{c} = \left(\mathbf{v}, \bm{\rho}\right)$, an encryption of $\bm{\mu}$ under pk.
\item \Decrypt$(\sk=(\mathbf{x}, \mathbf{y}), \mathbf{c}=(\mathbf{v}, \bm{\rho}))$: returns $\mathcal{C}$.\textsf{Decode}$(\bm{\rho}-\mathbf{v\cdot y})$.
\end{itemize}

Notice that the generator matrix $\mathbf{G}$ of the code $\mathcal{C}$ is publicly known, so the security of the scheme and the ability to decrypt do not rely on the knowledge of the error correcting code $\mathcal{C}$ being used.

\bigskip\noindent{\bf Correctness. } The correctness of our new encryption scheme clearly relies on the decoding capability of the code $\mathcal{C}$. Specifically, assuming $\mathcal{C}$.\textsf{Decode} correctly decodes $\bm{\rho}-\mathbf{v\cdot y}$, we have:
\begin{equation}
\Decrypt\left(\textnormal{sk}, \Encrypt\left(\textnormal{pk}, \bm{\mu}, \theta\right)\right) = \bm{\mu}.
\end{equation}
And $\mathcal{C}$.\textsf{Decode} correctly decodes $\bm{\rho} - \mathbf{x\cdot y}$ whenever
\begin{eqnarray}
%&&\omega(\bm{\rho}-\mathbf{v\cdot y}) \leq  \delta \\
&&\omega\left(\mathbf{s}\cdot \mathbf{r}_2 - \mathbf{v} \cdot \mathbf{y} + \bm{\epsilon}\right) \leq  \delta \\
&&\omega\left(\left(\mathbf{x}+\mathbf{q}_r \cdot \mathbf{y}\right) \cdot \mathbf{r}_2 - \left(\mathbf{r}_1 + \mathbf{q}_r \cdot \mathbf{r}_2 \right) \cdot \mathbf{y} + \bm{\epsilon}\right) \leq \delta \\
\label{eq:error-distro}&&\omega\left(\mathbf{x} \cdot \mathbf{r}_2-\mathbf{r}_1 \cdot \mathbf{y} + \bm{\epsilon}\right) \leq \delta
\end{eqnarray} 
In order to provide an upper bound on the decryption failure probability, an analysis of the distribution of the error vector $\mathbf{x\cdot r}_2-\mathbf{r}_1\cdot\mathbf{y} + \bm{\epsilon}$ is provided in Sec.~\ref{sec:5-analysis}.

\section{Security of the Scheme}
\label{sec:secu}
In this section we prove the security of our scheme, the proof is generic for any metric, and the security 
is reduced to the respective quasi-cyclic problems defined for Hamming and rank metric in Section 2.
\begin{theorem}
	The scheme presented above is \INDCPA under the $2$-\DQCSD and $3$-\DQCSD assumptions.
\end{theorem}

\begin{proof}
To prove the security of the scheme, we are going to build a sequence of games transitioning from an adversary receiving an encryption of message $\bm{\mu}_0$ to an adversary receiving an encryption of a message $\bm{\mu}_1$ and show that if the adversary manages to distinguish one from the other, then we can build a simulator breaking the \DQCSD assumption, for QC codes of order $2$ or $3$ (codes with parameters $[2n,n]$ or $[3n,2n]$), and running in approximately the same time.

\begin{games}
	\game{G:Real0}
This is the real game, we run an honest $\KeyGen$ algorithm, and after receiving $(\bm{\mu}_0, \bm{\mu}_1)$ from the adversary we produce an  encryption of $\bm{\mu}_0$.

	\game{G:SemId0}
In this game we start by forgetting the decryption key $\sk$, and taking ${\bf s}$ at random, and then proceed honestly.

	\game{G:Id00}
Now that we no longer know the decryption key, we can start generating random ciphertexts. So instead of picking correctly weighted $\mathbf{r}_1,\mathbf{r}_2,\bm{\epsilon}$, the simulator now picks random vectors in the full space.

\game{G:Id01}
We now encrypt the other plaintext. We chose $\mathbf{r}'_1,\mathbf{r}'_2,\bm{\epsilon}'$ uniformly and set $\mathbf{v}^\top = \mathbf{Q}\bm{r}'^\top$ and $\bm{\rho} = \bm{\mu}_1\mathbf{G} + \mathbf{s}\cdot \mathbf{r}'_2 + \bm{\epsilon}'$.

\game{G:Id10}
In this game, we now pick $\mathbf{r}'_1,\mathbf{r}'_2,\bm{\epsilon}'$ with the correct weight.

\game{G:Real1}
We now conclude by switching the public key to an honestly generated one.
\end{games}

The only difference between Game~\printgame{G:Real0} and Game~\printgame{G:SemId0} is the $s$ in the public key sent to the attacker at the beginning of the IND-CPA game. If the attacker has an algorithm $\mathcal{A}$ able to distinguish these two games he can build a distinguisher for the $\DQCSD$ problem. Indeed for a $\DQCSD$ challenge $(Q,s)$ he can: adjoin $G$ to build a public key;  run the IND-CPA game with this key and algorithm $\mathcal{A}$; decide on which Game he is. He then replies to the $\DQCSD$ challenge saying that $(Q,s)$ is uniform if he is on Game~\printgame{G:SemId0} or follows the $\QCSD$ distribution if he is in Game~\printgame{G:Real0}.

In both Game~\printgame{G:SemId0} and Game~\printgame{G:Id00}  the plaintext encrypted is known to be $\mu_0$ the attacker can compute:
$$\left(\begin{matrix}
\mathbf{v} \\ \bm{\rho} - \bm{\mu}_0 \mathbf{G}
\end{matrix}\right)
=
\left(\begin{matrix}
\mathbf{I}_n & \mathbf{0} & \textnormal{\textbf{rot}}(\mathbf{q}_r) \\
\mathbf{0} & \mathbf{I}_n & \textnormal{\textbf{rot}}(\mathbf{s})
\end{matrix}\right)
\cdot 
\left(
\mathbf{r}_1, \bm{\epsilon}, \mathbf{r}_2
\right)^\top
$$
The difference between Game~\printgame{G:SemId0} and Game~\printgame{G:Id00} is that in the former $(\mathbf{v}, \bm{\rho} - \bm{\mu}_0 \mathbf{G})$ follows the $\QCSD$ distribution (for a $2n \times 3n$ QC matrix of order $3$), and in the latter it follows a uniform distribution (as $\bm{r}_1$ and $\bm{\epsilon}$ are  uniformly distributed and independently chosen One-Time Pads). If the attacker is able to distinguish Game~\printgame{G:SemId0} and Game~\printgame{G:Id00} he can therefore break the $3-\DQCSD$ assumption.

The outputs from Game~\printgame{G:Id00} and Game~\printgame{G:Id01} follow the exact same distribution, and therefore the two games are indistinguishable from an information-theoretic point of view. Indeed, for each tuple $(\bm{r}, \bm{\epsilon})$ of Game~\printgame{G:Id00}, resulting in a given $(\mathbf{v},\bm{\rho})$,  there is a one to one mapping to a couple $(\bm{r}',\bm{\epsilon}')$ resulting in  Game~\printgame{G:Id01} in the \emph{same} $(\mathbf{v},\bm{\rho})$, namely $\bm{r}'=\bm{r}$ and $\bm{\epsilon}' - \bm{\mu}_0 \mathbf{G} +\bm{\mu}_1 \mathbf{G}$. This implies that choosing uniformly $(\bm{r}, \bm{\epsilon})$ in Game~\printgame{G:Id00} and choosing uniformly $(\bm{r}', \bm{\epsilon}')$ in Game~\printgame{G:Id01} leads to the same output distribution for $(\mathbf{v},\bm{\rho})$.

Game~\printgame{G:Id01} and Game~\printgame{G:Id10} are the equivalents of Game~\printgame{G:Id00} and Game~\printgame{G:SemId0} except $\bm{\mu}_1$ is used instead of $\bm{\mu}_0$. A distinguisher between these two games breaks therefore the $3-\DQCSD$ assumption too. Similarly Game~\printgame{G:Id01} and Game~\printgame{G:Real1} are the equivalents of Game~\printgame{G:SemId0} and Game~\printgame{G:Real0} and a distinguisher between these two games breaks the $\DQCSD$ assumption.

\smallskip

We managed to build a sequence of games allowing a simulator to transform a ciphertext of a message $\bm{\mu}_0$ to a ciphertext of a message $\bm{\mu}_1$. Hence the advantage of an adversary against the $\INDCPA$ experiment is bounded: 

\begin{equation}\label{eq:comprendPas}
\Adv^{\ind}_{\Ee,\Ae}(\seck) \leq  2 \cdot \left( \Adv^{2\textnormal{-}\DQCSD}(\seck) + \Adv^{3\textnormal{-}\DQCSD}(\seck) \right).
\end{equation}
%\qed
\end{proof}

\philippe{Preuve en métrique de hamming, mettre une note comme quoi ça marche aussi en rang en changeant les noms des schémas}

\section{Analysis of the Distribution of the Error Vector of the Scheme for Hamming Distance}
\label{sec:5-analysis}

The aim of this Section is to determine the probability that 
the condition in Eq. (\ref{eq:error-distro}) holds. In order to 
do so, we study the error distribution of the error vector 
$\mathbf{e} = \mathbf{x\cdot r}_2 - \mathbf{r}_1\cdot \mathbf{y} + \bm{\epsilon}$.
% In the following, we denote by $\mathbb{E}_{\omega}[n,w,\epsilon]$, 
% the expected weight of the vector $\mathbf{e}$ defined before. 
% Using the linearity of Expectation we have that the expected weight 
% of the error vector is simply $n$ times the expected weight $E$ of a 
% coordinate. From that we have $\mathbb{E}_w[n,w,\epsilon] = n \cdot E$, 
% and this allows to focus on a single coordinate in our analysis.

%\ob{C'est l'argument a la truelle...}

\medskip

%Now, as mentioned in Sec.~\ref{sec:1-3-overview}, 
The vectors $\mathbf{x},\mathbf{y},\mathbf{r}_1,\mathbf{r_2},
\bm{\epsilon}$ 
have been taken to be uniformly and independently chosen among vectors
of weight $w$. A very close probabilistic model is when all these
independent vectors are chosen to follow the distribution of random
vectors whose coordinates are independent Bernoulli variables of
parameter $p=w/n$. To simplify analysis we shall assume this model
rather than the constant weight uniform model. Both models are very
close, and our cryptographic protocols work just as well in both settings.

We first evaluate the distributions of the products $\mathbf{x\cdot
  r}_2$ and $\mathbf{r}_1\cdot \mathbf{y}$.

\begin{proposition}\label{pr:prod}
Let $\mathbf{x}=(X_1,\ldots,X_n)$ be a random vector where the $X_i$
are independent Bernoulli variables of parameter $p$,
$P(X_i=1)=p$. Let $\mathbf{y}=(Y_1,\ldots,Y_n)$ be a vector following
the same distribution and independent of $\mathbf{x}$.
Let $\mathbf{z = x\cdot y}=(Z_1,\ldots,Z_n)$ as defined in Eq. (\ref{eq:product}). Then 
\begin{equation}\label{eq:ptilde}
    \begin{cases}
\mathrm{Pr}[Z_k=1] =
\displaystyle    
\sum_{0\leq i\leq n,\atop i \textnormal{ odd}}
\binom{n}{i}p^{2i}\left(1-p^2\right)^{n-i}, \\
\mathrm{Pr}[Z_k=0] =
\displaystyle\sum_{0\leq i\leq n,\atop i \textnormal{ even}}
\binom{n}{i}p^{2i}\left(1-p^2\right)^{n-i}.
\end{cases}
\end{equation}
\end{proposition}

\begin{proof}
  We have 
  \begin{equation}
    \label{eq:Zk}
    Z_k=\sum_{i+j=k\bmod n}X_iY_j \quad \bmod 2.
  \end{equation}
  Every term $X_iY_j$ is the product of two independent Bernoulli
  variables of parameter $p$, and is therefore a Bernoulli variable of
  parameter $p^2$. The variable $Z_k$ is the sum of $n$ such products,
  which are all independent since every variable $X_i$ is involved
  exactly once in \eqref{eq:Zk}, for $0\leq i\leq n-1$, and similarly
  every variable $Y_j$ is involved once in \eqref{eq:Zk}.
  Therefore $Z_k$ is the sum modulo $2$ of $n$ independent Bernoulli
  variables of parameter $p^2$.
\end{proof}

Let us denote by $\tilde{p} = \tilde{p}(n, w) = \mathrm{Pr}[z_k = 1]$
from Eq. \eqref{eq:ptilde}. We will be working in the regime where
$w=\omega\sqrt{n}$, meaning $p^2=(\frac wn)^2=\omega^2/n$. When $n$
goes to infinity we have that the binomial distribution of the weight of the binary $n$-tuple
  $$(X_iX_j)_{i+j=k\,\bmod n}$$
converges to the Poisson distribution of parameter $\omega^2$ so that, for
fixed $\omega=w/\sqrt{n}$,
\begin{equation}
  \label{eq:poisson}
  \tilde{p}(n, w)=\mathrm{Pr}[z_k = 1]\xrightarrow[n\to\infty]{}
  e^{-\omega^2}\sum_{\ell\,\text{odd}}\frac{\omega^{2\ell}}{\ell !} = e^{-\omega^2}\sinh\omega^2.
\end{equation}

Let $\mathbf{x, y, r}_1, \mathbf{r}_2$ be independent random vectors
whose coordinates are independently Bernoulli distributed with
parameter $p$. Then the $k$-th coordinates of $\mathbf{x\cdot r}_2$
and of $\mathbf{r}_1\cdot \mathbf{y}$ are independent and Bernoulli
distributed with parameter $\tilde{p}$. Therefore their  modulo $2$ sum
$\mathbf{t} = \mathbf{x\cdot r}_2 - \mathbf{r}_1\cdot \mathbf{y}$ is
Bernoulli distributed with
\begin{equation}\label{eq:sum-prod}
    \begin{cases}
     \mathrm{Pr}[t_k=1] =
      2\tilde{p}(1-\tilde{p}), \\
     \mathrm{Pr}[t_k=0] =
(1-\tilde{p})^2+\tilde{p}^2 .
    \end{cases}
\end{equation}

Finally, by adding the final term $\bm{\epsilon}$ to $\mathbf{t}$, 
we obtain the distribution of the coordinates of the error vector
$\mathbf{e} = \mathbf{x\cdot r}_2 - \mathbf{r}_1\cdot \mathbf{y} +
\bm{\epsilon}$. Since the coordinates of $\bm{\epsilon}$ are Bernoulli
of parameter $p$ and those of
$\mathbf{t}$ are Bernoulli distributed as
\eqref{eq:sum-prod} and independent from $\bm{\epsilon}$, we obtain~:

\begin{theorem}\label{th:error-distro}
Let $\mathbf{x, y, r}_1, \mathbf{r}_2 \sim \mathcal{B}\left(n,\frac{w}{n}\right)$, $\bm{\epsilon} \sim \mathcal{B}\left(n, \epsilon\right)$, and let $\mathbf{e} = \mathbf{x\cdot r}_2 - \mathbf{r}_1\cdot \mathbf{y} + \bm{\epsilon}$. Then
\begin{equation}
\label{eq:pr-total}
    \begin{cases}
\mathrm{Pr}[e_k=1] = 
2\tilde{p}(1-\tilde{p})(1-\frac{\epsilon}{n}) +
\left((1-\tilde{p})^2+\tilde{p}^2\right)\frac{\epsilon}{n}, \\
\mathrm{Pr}[e_k=0] = \left((1-\tilde{p})^2+\tilde{p}^2\right)(1-\frac{\epsilon}{n}) +
2\tilde{p}(1-\tilde{p})\frac{\epsilon}{n}.
    \end{cases}
\end{equation}
\end{theorem}

Theorem~\ref{th:error-distro} gives us the probability that a
coordinate of the error vector $\mathbf{e}$ is $1$. 
In our simulations to follow, which occur in the regime $p=\omega\sqrt{n}$ with constant
$\omega$, we make the simplifying assumption that the coordinates
of $\mathbf{e}$ are independent, meaning that the weight of $\mathbf{e}$
follows a binomial distribution of parameter $p^\star$,
where $p^\star$ is defined as in Eq. (\ref{eq:pr-total}): $p^\star =
p^\star (n, w) = 2\tilde{p}(1-\tilde{p})(1-\frac{\epsilon}{n}) +
\left((1-\tilde{p})^2+\tilde{p}^2\right)\frac{\epsilon}{n}$. This
approximation will give us, for $0\leq d \leq \min(2w^2 +\epsilon, n)$,
\begin{equation}\label{eq:weight}
\mathrm{Pr}[\omega(\mathbf{e}) = d] = \binom{n}{d}{(p^\star)}^d{(1-p^\star)}^{(n-d)}.
\end{equation}
In practice, the results obtained by simulation on the decryption failure are very coherent with this assumption.

% We now make the simplifying assumption that for our considered parameters ($w = \mathcal{O}\left(\sqrt{n}\right)$), the error vector $\mathbf{e} = \mathbf{x\cdot r}_2 - \mathbf{r}_1\cdot \mathbf{y} + \bm{\epsilon}$ is the result of $n$ successive independent samples of a random variable following a Bernoulli distribution of parameter $p^\star$, where $p^\star$ is defined as in Eq. (\ref{eq:pr-total}): $p^\star = p^\star (n, w) = 2\tilde{p}(1-\tilde{p})(1-\frac{\epsilon}{n}) + \left((1-\tilde{p})^2+\tilde{p}^2\right)\frac{\epsilon}{n}$ (notice that in practice the results obtained by simulation on the decryption failure are very coherent with this assumption). Under the assumption made above, the Hamming weight of the error vector follows a binomial distribution $\mathcal{B}(n, p^\star)$. Formally, for $0\leq d \leq \min(2\cdot w^2 +\epsilon, n)$, we have that:

%We provide two related algorithms (implemented in Sage~\cite{sage}) in the supplementary materials: \texttt{probaWeightGreaterThan} and \texttt{findWeightForProbaLessThan}. The first one takes as input the dimension $n$ of the code $\mathcal{C}$, the weight $w$ of vectors $\mathbf{x, y, r}_1$, and $\mathbf{r}_2$, and an upper bound $k$ on the weight of the error vector and returns $\mathrm{Pr}[\omega(\mathbf{e})>k]$. The latter one takes $n$, $w$, and a security parameter $\lambda$ as input and returns~|~assuming it exists~|~the smallest upper bound $k$ such that $\mathrm{Pr}[\omega(\mathbf{e})>k] < 2^{-\lambda}$.

%\philippe{Appendix~\ref{proba}}
%
%\philippe{En intro mettre appendix C.1}

\section{Decoding Codes with Low Rates and Good Decoding Properties}
\label{sec:6-low-rate}

The previous Section allowed us to determine the distribution of the error vector $\mathbf{e}$ in the configuration where a simple linear code is used. Now the decryption part corresponds to decoding the error described in the previous section. Any decodable code
can be used at this point, depending on the considered application:
clearly small dimension codes will allow better decoding, 
but at the cost of a lower encryption rate. The particular case that we consider corresponds typically
to the case of key exchange or authentication, where only a small amount
of data needs to be encrypted (typically 80, 128 or 256 bits, a symmetric secret key size). 
We therefore need codes with low rates which are able to correct many errors.
Again, a tradeoff is necessary between efficiently decodable codes but with a high decoding cost and less efficiently decodable codes
but with a smaller decoding cost.

An example of such a family of codes with good decoding properties, 
meaning a simle decoding algorithm which can be analyzed,
is given by Tensor Product Codes, which are used for biometry~\cite{BCC07}, where the same type
of issue appears. More specifically, we will consider a special simple case of Tensor Product Codes (BCH codes and
repetition codes), for which a precise analysis of the decryption
failure can be obtained in the Hamming distance case.

\subsection{Tensor Product Codes}
\label{sec:6-1-tensor}
%\philippe{On va décoder des codes de taux faibles avec beaucoup d'erreurs. Pour pouvoir faire l'analyse on utilise des codes tensorisés avec répétition}

%\philippe{Appendix C moins la partie d'Olivier}

\begin{definition}[Tensor Product Code]\label{def:tensor}
Let $\mathcal{C}_1$ (resp. $\mathcal{C}_2$) be a $[n_1, k_1, d_1]$ (resp. $[n_2, k_2, d_2]$) linear code over $\mathbb{F}$. The \emph{Tensor Product Code} of $\mathcal{C}_1$ and $\mathcal{C}_2$ denoted $\mathcal{C}_1 \otimes \mathcal{C}_2$ is defined as the set of all $n_2 \times n_1$ matrices whose rows are codewords of $\mathcal{C}_1$ and whose columns are codewords of $\mathcal{C}_2$.

More formally, if $\mathcal{C}_1$ (resp. $\mathcal{C}_2$) is generated by $\mathbf{G}_1$ (resp. $\mathbf{G}_2$), then \begin{equation}
\mathcal{C}_1 \otimes \mathcal{C}_2 = \left\lbrace\mathbf{G}_2^\top\mathbf{XG}_1\textnormal{ for }\mathbf{X} \in \mathbb{F}^{k_2 \times k_1} \right\rbrace
\end{equation}
\end{definition}

\begin{remark}
Using the notation of the above Definition, the tensor product of two linear codes is a $[n_1n_2, k_1k_2, d_1d_2]$ linear code.
\end{remark}

\subsection{Specifying the Tensor Product Code}
\label{sec:bch}
Even if 
tensor product codes seem well-suited for our purpose, an analysis 
similar to the one in Sec.~\ref{sec:5-analysis} becomes much more 
complicated. Therefore, in order to provide strong guarantees on 
the decryption failure probability for our cryptosystem, we chose 
to restrict ourselves to a tensor product code 
$\mathcal{C} = \mathcal{C}_1 \otimes \mathcal{C}_2$, where 
$\mathcal{C}_1$ is a BCH$(n_1, k, \delta_1)$ code of length $n_1$, 
dimension $k$, and correcting capability $\delta_1$ (\emph{i.e.} it 
can correct up to $\delta_1$ errors), and $\mathcal{C}_2$ is the 
repetition code of length $n_2$ and dimension $1$, denoted 
$\mathds{1}_{n_2}$. (Notice that $\mathds{1}_{n_2}$ can decode up to
 $\delta_2 = \lfloor \frac{n_2-1}{2} \rfloor$.) Subsequently, the 
analysis becomes possible and remains accurate but the negative 
counterpart is that there probably are some other tensor product 
codes achieving better efficiency (or smaller key sizes). 

%In Sec.~\ref{sec:5-analysis}, we assumed $\bm{\epsilon} \lar \mathcal{V}$ had weight $\omega(\bm{\epsilon}) = w$ as $\mathbf{x, y, r}_1$ and $\mathbf{r}_2$. In the following, we assume $\bm{\epsilon}$ can be of any weight $0 \leq \epsilon \leq n$. The probability $p^\star$ of the coordinate $e_k$ being $1$ hence becomes: 
%\begin{equation}
%p^\star = p^\star (n, w, \epsilon) = 2\tilde{p}(1-\tilde{p})(1-\frac{\epsilon}{n}) + \left((1-\tilde{p})^2+\tilde{p}^2\right)\frac{\epsilon}{n}.
%\end{equation}

\medskip In the Hamming metric version of the cryptosystem we propose, 
a message $\bm{\mu} \in \mathbb{F}^{k}$ is first encoded into 
$\bm{\mu}_1 \in \mathbb{F}^{n_1}$ with a BCH($n_1, k_1 = k, \delta_1$) 
code, then each coordinate $\mu_{1,i}$ of $\bm{\mu}_1$ is re-encoded 
into $\tilde{\bm{\mu}}_{1,i} \in \mathbb{F}^{n_2}$ with a repetition 
code $\mathds{1}_{n_2}$. We denote $n=n_1n_2$ the length of the 
tensor product code (its dimension is $k =k_1\times 1$), and by 
$\tilde{\bm{\mu}}$ the resulting encoded vector, \emph{i.e.} 
$\tilde{\bm{\mu}} = \left(\tilde{\bm{\mu}}_{1,1}, \ldots, \tilde{\bm{\mu}}_{1,n_1}\right) \in \mathbb{F}^{n_1n_2}$.
%This encoding method is known as a tensor product code in the literature \jc{citations}, and gives rise to the following definition:

%For our scheme instantiated with Hamming Codes, we will have $\mathcal{C} = \mathcal{C}_1 \otimes \mathcal{C}_2$ where $\mathcal{C}_1$ is a \textbf{BCH}($n_1, k_1 = k, d_1$) code and $\mathcal{C}_2 = \mathds{1}_{n_2}$ the $[n_2, k_2 = 1, d_2 = n_2]$ repetition code, that can decode up to $\delta_2 = \lfloor \frac{n_2-1}{2} \rfloor$.

The efficient algorithm used for the repetition code is the majority decoding, \emph{i.e.} more formally: \begin{equation}
\mathds{1}_{n_2}.\textsf{Decode}(\tilde{\bm{\mu}}_{1,j}) = \left\{
                \begin{array}{ll}
1\textnormal{ if } \sum_{i=0}^{n_2-1}\tilde{\mu}_{1,j,i} \geq \lceil \frac{n_2+1}{2} \rceil, \\
0\textnormal{ otherwise.}
                \end{array}
              \right. 
\end{equation}

\bigskip\noindent{\bf Decryption Failure Probability. } With a tensor 
product code $\mathcal{C} = \mathrm{BCH}(n_1, k, \delta) \otimes \mathds{1}_{n_2}$
as defined above, a decryption failure occurs whenever the decoding 
algorithm of the BCH code does not succeed in correcting errors 
that would have arisen after wrong decodings by the repetition code.
Therefore, the analysis of the decryption failure probability is again 
split into three steps: evaluating the probability that the repetition code 
does not decode correctly, the conditional probability of a wrong 
decoding for the BCH code given an error weight and finally, the 
decryption failure probability using the law of total probability.

\bigskip\noindent{\bf Step 1. } We now focus on the probability that 
an error occurs while decoding the repetition code. As shown in 
Sec.~\ref{sec:5-analysis}, the probability for a coordinate of 
$\mathbf{e = x\cdot r}_2 - \mathbf{r}_1\cdot \mathbf{y} + \bm{\epsilon}$ to be 
$1$ is $p^\star = p^\star(n_1n_2, w, \epsilon)$ (see Eq. (\ref{eq:pr-total})).
As mentioned above, $\mathds{1}_{n_2}$ can decode up to 
$\delta_2 = \lfloor \frac{n_2 -1}{2} \rfloor$ errors. Therefore, assuming 
that the error vector $\mathbf{e}$ has weight $\gamma$ (which occurs with the 
probability given in Eq. (\ref{eq:weight})), the probability of getting a 
decoding error on a single block of the repetition code $\mathds{1}_{n_2}$ 
is hence given by:
\begin{equation}
\label{eq:pbar}
\bar{p}_\gamma = \bar{p}_\gamma(n_1, n_2) = \sum_{i = \lfloor\frac{n_2-1}{2}\rfloor + 1}^{n_2} \binom{n_2}{i} \left(\frac{\gamma}{n_1n_2}\right)^i \left(1-\frac{\gamma}{n_1n_2}\right)^{n_2-i}.
\end{equation}
%%% OR EQUIVALENTLY
%\begin{equation}
%\bar{p} = \bar{p}(n_1, n_2, \gamma) = \sum_{i = \lfloor\frac{n_2-1}{2}\rfloor + 1}^{n_2} \binom{n_2}{i} \left(p^\star\right)^i \left(1-p^\star\right)^{n_2-i}.
%\end{equation}

\bigskip\noindent{\bf Step 2. } We now focus on the BCH$(n_1, k, \delta_1)$ code, and recall that it can correct up to $\delta_1$ errors. Now the probability $\mathcal{P}$ that the BCH$(n_1, k, \delta_1)$ code fails to decode correctly the encoded message $\bm{\mu}_1$ back to $\bm{\mu}$ is given by the probability that an error occurred on at least $\delta_1 +1$ blocks of the repetition code. Therefore, we have
\begin{equation}
\label{eq:pcapital}
\mathcal{P} = \mathcal{P}(\delta_1, n_1, n_2, \gamma) = \sum_{i=\delta_1 + 1}^{n_1} \binom{n_1}{i}\left(\bar{p}_\gamma\right)^i(1-\bar{p}_\gamma)^{n_1-i}.
\end{equation}

%Finally, for (pk, sk)$\lar$ \textsf{G.KeyGen}$(n, k, w, \delta, \epsilon)$, the total probability of getting a decryption failure can be given simply by:
%\marginnote{Rajouter $\epsilon$ dans \textsf{G.Enc}() et \textsf{G.KeyGen}()} 

\bigskip\noindent{\bf Step 3. } Finally, using the law of total 
probability, we have that the decryption failure 
probability is given by the sum\ul{,} over all the possible weights\ul{,} of the 
probability that the error has this specific weight times the 
probability of a decoding error for this weight. This is captured in the 
following theorem, whose proof is a straightforward consequence of 
the formulae of Sec.~\ref{sec:5-analysis} and~\ref{sec:6-1-tensor}.

\begin{theorem}
Let $\mathcal{C} = \mathrm{BCH}(n_1, k, \delta) \otimes \mathds{1}_{n_2}$, $(\textsf{pk, sk}) \leftarrow \KeyGen$, $\bm{\mu} \lar \mathbb{F}_2^k$, and some randomness $\theta \in \{0,1\}^*$, then with the notations above, the decryption failure probability is
\end{theorem}
\begin{eqnarray}
p_{\textnormal{fail}} &=& \mathrm{Pr}[\textnormal{\Decrypt}\left(\textnormal{sk}, \Encrypt\left(\textnormal{pk}, \bm{\mu}, \theta \right)\right) \neq \bm{\mu}.
]\\
\label{eq:pfail}&=& \sum_{\gamma=0}^{\min\left(2w^2+\epsilon, n_1n_2\right)} \mathrm{Pr[\omega(\mathbf{e}) = \gamma]} \cdot \mathcal{P}(\delta_1, n_1, n_2, \gamma)
\end{eqnarray}
%\marginnote{Changer \textsf{Enc} en \textsf{Encrypt} ?}
%\jc{Je glisserais peut-être un mot sur la sous-multiplicativité de $\omega(\cdot)$ pour expliquer le $2w^2+\epsilon$...}

\section{Parameters}
\label{sec:7-parameters}
\subsection{HQC Instantiation for Hamming Metric}\label{sec:7-1-hamming}

In this Section, we describe our new cryptosystem in the Hamming 
metric setting. As mentioned in the previous Section, we use a 
tensor product code (Def.~\ref{def:tensor}) 
$\mathcal{C} = \mathrm{BCH}(n_1, k, \delta) \otimes \mathds{1}_{n_2}$. 
A message $\bm{\mu} \in \mathbb{F}^{k}$ is encoded into 
$\bm{\mu}_1 \in \mathbb{F}^{n_1}$ with the BCH code, then each 
coordinate $\mu_{1,i}$ of $\bm{\mu}_1$ is encoded into 
$\tilde{\bm{\mu}}_{1,i} \in \mathbb{F}^{n_2}$ with 
$\mathds{1}_{n_2}$. To match the description of our cryptosystem 
in Sec.~\ref{sec:3-2-presentation}, we have 
$\bm{\mu}\mathbf{G} = \tilde{\bm{\mu}} = \left(\tilde{\bm{\mu}}_{1,1}, \ldots, \tilde{\bm{\mu}}_{1,n_1}\right) \in \mathbb{F}^{n_1n_2}$. To obtain the ciphertext, 
$\mathbf{r = (r}_1, \mathbf{r}_2) \lar \mathcal{V}^2$ and $\bm{\epsilon} \lar \mathcal{V}$ 
are generated and the encryption of $\bm{\mu}$ is 
$\mathbf{c} = (\mathbf{rQ}^\top, \bm{\rho} = \bm{\mu}\mathbf{G} + \mathbf{s\cdot r}_2 + \bm{\epsilon})$.

\bigskip\noindent{\bf Parameters for Our Scheme. } 
%We provide two 
%kinds of parameter sets depending on the decryption failure 
%probability one is ready to deal with; the higher this probability, 
%the more efficient the cryptosystem. 
 We provide two sets of parameters: the first one in Tab.~\ref{tab:params_h_pre} targets different
pre-quantum security levels while the second one in Tab.~\ref{tab:params_h_post} is quantum-safe.
For each parameter set, the 
parameters are chosen so that the minimal workfactor of the best 
known attack exceeds the security parameter. For classical 
attacks, best known attacks include the works from~\cite{CanCha98,BLP08,AC:FinSen09,EC:BJMM12}
and for quantum attacks, the work of~\cite{Bernstein10}.
We consider $w = \mathcal{O}\left(\sqrt{n}\right)$ and follow the complexity described in~\cite{CTS16}.

Note that our cryptosystem is quite efficient since the decryption simply involves a decoding of a repetition code and a small length BCH code.

\begin{center}
\begin{table}[h]
\begin{center}
%\resizebox{\textwidth}{!}{
\begin{tabular}{cccccccccc}  
\toprule
& \multicolumn{9}{c}{Cryptosystem Parameters} \\
\cmidrule(r){2-10}
~~Instance~~ & $~n_1~$ & $~n_2~$ & $n_1n_2 = n~$ & $~k~$ & $~\delta~$ & $~w~$ & $~\epsilon=3w~$ & ~security~ & $~ p_\textnormal{fail}~$ \\
\midrule
Toy & $~255~$ & $~25~$ & $~6,379~$ & $~63~$ & $~30~$ & $~36~$ & $~108~$ & $~64~$ & $<2^{-64}$\\
Low & $~255~$ & $~37~$ & $~9,437~$ & $~79~$ & $~27~$ & $~45~$ & $~135~$ & $~80~$ & $<2^{-80}$\\
Medium & $~255~$ & $~53~$ & $~13,523~$ & $~99~$ & $~23~$ & $~56~$ & $~168~$ & $~100~$ & $<2^{-100}$\\
Strong & $~511~$ & $~41~$ & $~20,959~$ & $~121~$ & $~58~$ & $~72~$ & $~216~$ & $~128~$ & $<2^{-128}$ \\
\bottomrule \\[1mm]
\end{tabular}
%}
\caption{\label{tab:params_h_pre}Parameter sets for our cryptosystem in Hamming metric. The tensor product code used is $\mathcal{C} = \textnormal{BCH}(n_1, k, \delta) \otimes \mathds{1}_{n_2}$. The parameters for the BCH codes were taken from~\cite{codeBookPetWel72}. Security in the first four instances is given in bits, in the classical model of computing. In the last four instances, the security level is the equivalent of the classical security level but in the quantum computing model, following the work of~\cite{Bernstein10}. The public key size, consisting of $(\mathbf{q}_r, \mathbf{x}+\mathbf{q}_r\cdot \mathbf{y})$, has size $2n$ (in bits) (although considering
a seed for $\mathbf{q}_r$ the size can be reduced to $n$ plus the size of the seed), and the secret key (consisting of $\mathbf{x}$ and $\mathbf{y}$ both of weight $w$) has size $2w\lceil\log_2(n)\rceil$ (bits) - which again can be reduced to the size of a seed. Finally, the size of
the encrypted message is $2n$.}
\end{center}
\end{table}
\end{center}

\bigskip\noindent{\bf Specific structural attacks. } Quasi-cyclic codes have a special structure
which may potentially open the door to specific structural attacks. Such attacks have been studied in \cite{GJL15,LJK+16,S11}, these attacks are especially efficient in the case when the polynomial $x^n-1$ has many small factors. These attacks become inefficient as soon as $x^n-1$ has only two
factors of the form $(x-1)$ and $x^{n-1}+x^{n-2}+...+x+1$, which is the case when $n$ is primitif in $\mathbb{F}_q$, for $q=2$ it corresponds to cases when $2$ generates $\left(\mathbb{Z}/n\mathbb{Z}\right)^*$, such numbers are known up to very large values. We consider such $n$ for our parameters.
\medskip
In Tab.~\ref{tab:params_h_pre} and~\ref{tab:params_h_post}, $n_1$ denotes the length of the BCH code, $n_2$ the length of the repetition code $\mathds{1}$ so that the length of the tensor product code $\mathcal{C}$ is $n = n_1n_2$ (actually the smallest primitive prime greater than $n_1n_2$). $k$ is the dimension of the BCH code and hence also the dimension of $\mathcal{C}$. $\delta$ is the decoding capability of the BCH code, \emph{i.e.} the maximum number of errors that the BCH can decode. $w$ is the weight of the $n$-dimensional vectors $\mathbf{x}$, $\mathbf{y}$, $\mathbf{r}_1$, and $\mathbf{r}_2$ and similarly $\epsilon = \omega(\bm{\epsilon}) = 3\times w$ for our cryptosystem.

\begin{center}
\begin{table}
\begin{center}
%\resizebox{\textwidth}{!}{
\begin{tabular}{cccccccccc}  
\toprule
& \multicolumn{9}{c}{Cryptosystem Parameters} \\
\cmidrule(r){2-10}
~~Instance~~ & $~n_1~$ & $~n_2~$ & $n_1n_2 = n~$ & $~k~$ & $~\delta~$ & $~w~$ & $~\epsilon=3w~$ & ~security~ & $~ p_\textnormal{fail}~$ \\
\midrule
Toy & $~255~$ & $~65~$ & $~16,603~$ & $~63~$ & $~87~$ & $~72~$ & $~216~$ & $~64~$ & $<2^{-64}$\\
Low & $~511~$ & $~47~$ & $~24,019~$ & $~76~$ & $~85~$ & $~89~$ & $~267~$ & $~80~$ & $<2^{-80}$\\
Medium & $~255~$ & $~141~$ & $~35,963~$ & $~99~$ & $~23~$ & $~112~$ & $~336~$ & $~100~$ & $<2^{-100}$\\
Strong & $~511~$ & $~109~$ & $~55,711~$ & $~121~$ & $~58~$ & $~143~$ & $~429~$ & $~128~$ & $<2^{-128}$ \\
\bottomrule \\ % TODO : Once done, check whether BCH(1023, ..., ...) improve n
\end{tabular}
%}
\caption{\label{tab:params_h_post}Parameters for quantum-safe HQC. All parameters are similar to Tab.~\ref{tab:params_h_pre}.}
\end{center}
\end{table}
\end{center}
%n1   n2   n=n1n2   k   delta   w   epsilon   securité   proba_echec

\bigskip\noindent{\bf Computational Cost. } The most expensive part of the encryption and decryption is the matrix vector product, in practice
the complexity is hence $\mathcal{O}(n^{\frac{3}{2}})$ (for $w = \mathcal{O}(\sqrt{n})$). Asymptotically the cost becomes linear in $n$.

Notice that it would be possible to consider other types of decodable codes in order to increase the encryption rate to $1/4$ (say),
but at the cost of an increase of the length of the code, for instance using LDPC (3,6) codes would increase the rate,
but multiply the length by a factor of roughly three.

\subsection{RQC Instantiation for Rank Metric}\label{sec:7-2-rank}
\noindent{\bf Error distribution and decoding algorithm: no decryption failure. } The case of the rank metric is much more simpler than for Hamming metric.
Indeed in that case the decryption algorithm of our cryptosystem asks to decode an error 
$\mathbf{e}=\mathbf{x\cdot r}_2-\mathbf{r}_1\cdot \mathbf{y}+\bm{\epsilon}$ where the words $(\mathbf{x,y})$ and $(\mathbf{r}_1,\mathbf{r}_2)$ have rank
weight $w$. At the difference of Hamming metric the rank weight of the vector
$\mathbf{x\cdot r}_2-\mathbf{r}_1\cdot \mathbf{y}$ is almost always $w^2$ and is in any case bounded above by $w^2$.
In particular with a strong probability the rank weight of $\mathbf{x\cdot r}_2-\mathbf{r}_1\cdot \mathbf{y}$ is
the same than the rank weight of $\mathbf{x\cdot r}_2$ since $\mathbf{x}$ and $\mathbf{y}$ %on one side
share the same rank support, so as $\mathbf{r}_1$ and $\mathbf{r}_2$.
Hence for decoding, we consider Gabidulin $[n,k]$ codes over $\mathbb{F}_{q^{n}}$, which
can decode $\frac{n-k}{2}$ rank errors and choose our parameters such
that $w^2+\epsilon \le \frac{n-k}{2}$, so that, unlike the Hamming
metric case, there is no decryption failure.

\bigskip\noindent{\bf Parameters for Our Scheme. } In Tab.~\ref{tab:params_r_pre} and~\ref{tab:params_r_post}, $n$ denotes the length of the Rank metric code, $k$ its dimension, $q$ is the number of elements in the base field $\mathbb{F}_q$, and $m$ is the degree of the extension. Similarly to the Hamming instantiation, $w$ is the rank weight of vectors $\mathbf{x}$, $\mathbf{y}$, $\mathbf{r}_1$, and $\mathbf{r}_2$, and $\epsilon$ the rank weight of $\bm{\epsilon}$. %For the case of rank metric, we always consider $n'=n=m$.

\bigskip\noindent{\bf Specific structural attacks. } Specific attacks were described in \cite{HT15,AFRICACRYPT:GRSZ14} for LRPC cyclic codes. These attack
use the fact that the targeted code has a generator matrix formed from shifted low weight codewords and in the case
of \cite{HT15}, also uses multi-factor factorization of $x^n-1$. These attack corresponds to searching for low weight codewords of a given code of rate 1/2. In the present case the attacker has to search for a low weight word associated
to a non null syndrom, such that previous attacks imply considering a code with a larger dimension so that 
in practice these attacks do no improve on direct attacks on the syndrome.
Meanwhile in practice by default, we choose $n$ a primitive prime number, such that the polynomial $x^n-1$ has no factor of degree less than $\frac{n-1}{2}$ except $x-1$. The best attacks consists in decoding a random double-circulant
$[2n,n]$ over $\mathbb{F}_{q^m}$ for rank weight $\omega$. 

\bigskip

Examples of parameters are given in Tab.~\ref{tab:params_r_pre}
according to best known attacks (combinatorial attacks in practice) described in Sec.~\ref{sec:2-4-attacks}.
Quantum-safe parameters for RQC are given in Tab.~\ref{tab:params_r_post}. 
For the case of rank metric, we always consider $n'=n=m$.
%
%\begin{table}
%\begin{center}
%
%\begin{tabular}{|c|c|c|c|c|c|c|c|c|c|}
%\hline
%n & k' & m & q & $\omega$ & $w(\epsilon)$ & message size (bits)& key size (bits) & security (bits) & quantum security \\
%\hline
%37 & 13  & 37 &4 & 3 & 3  & 962 & 2738 & 90 & 45 \\
%\hline
%53 & 13 & 53 & 2 & 4 & 4 & 689 & 2809 & 95 & 47 \\
%\hline
%61 & 3 & 61 & 2 & 5 & 4 & 183 & 3721 & 140 & 70 \\
%\hline
%83 & 3 & 83 & 2 & 6 & 4 & 249 & 6889 & 230 & 115 \\
%\hline
%61 & 3 & 61 & 4 & 5 & 4 & 366 & 7442 & 264 & 132 \\
%\hline
%
%\end{tabular}
%\end{center}
%\caption{\label{tab:paramsRank}Parameter sets for rank metric}
%\end{table}

\begin{center}
\begin{table}
\begin{center}
\begin{tabular}{r@{\hspace{.20cm}}ccccccccc}  
\toprule
\multicolumn{10}{c}{Cryptosystem Parameters} \\
\cmidrule(r){2-10}
Instance & $~n~$ & $~k~$ & $~m~$ & $~q~$ & $~w~$ & $~\epsilon~$ & ~plaintext~ & ~key size~ & ~security~\\
\midrule
%RQC-I & $~37~$ & $~13~$ & $~37~$ & $~4~$ & $~3~$ & $~3~$ & $~962~$ & $~2,738~$ & $~90~$ & $~45~$\\
RQC-I & $~53~$ & $~13~$ & $~53~$ & $~2~$ & $~4~$ & $~4~$ & $~689~$ & $~2,809~$ & $~95~$\\
RQC-II & $~61~$ & $~3~$ & $~61~$ & $~2~$ & $~5~$ & $~4~$ & $~183~$ & $~3,721~$ & $~140~$\\
RQC-III & $~83~$ & $~3~$ & $~83~$ & $~2~$ & $~6~$ & $~4~$ & $~249~$ & $~6,889~$ & $~230~$\\
%RQC-V & $~61~$ & $~3~$ & $~61~$ & $~4~$ & $~5~$ & $~4~$ & $~366~$ & $~7,442~$ & $~264~$ & $~132~$\\
\bottomrule \\ % TODO : Once done, check whether BCH(1023, ..., ...) improve n
\end{tabular}
\caption{\label{tab:params_r_pre}Parameter sets for RQC: our cryptosystem in Rank metric. The plaintexts, key sizes, and security are expressed in bits.}
\end{center}
\end{table}
\end{center}

\noindent {\bf Remark.} The system is based on cyclic codes, which means considering polynomials modulo $x^n-1$,
interestingly enough, and only in the case of the rank metric, the construction remains valid
when considering not only polynomials modulo $x^n-1$ but also modulo a polynomial with coefficient
in the base field $GF(q)$. Indeed in that case the modulo does not change the rank weight of a codeword.
Such a variation on the scheme may be interesting to avoid potential structural attacks which 
may use the factorization of the quotient polynomial for the considered polynomial ring. 

\begin{center}
\begin{table}
\begin{center}
\begin{tabular}{r@{\hspace{.20cm}}ccccccccc}  
\toprule
\multicolumn{10}{c}{Cryptosystem Parameters} \\
\cmidrule(r){2-10}
Instance & $~n~$ & $~k~$ & $~m~$ & $~q~$ & $~w~$ & $~\epsilon~$ & ~plaintext~ & ~key size~ & ~security~ \\
\midrule
%RQC-I & $~37~$ & $~13~$ & $~37~$ & $~4~$ & $~3~$ & $~3~$ & $~962~$ & $~2,738~$ & $~90~$ & $~45~$\\
%RQC-I & $~53~$ & $~13~$ & $~53~$ & $~2~$ & $~4~$ & $~4~$ & $~689~$ & $~2,809~$ & $~95~$ & $~47~$\\
RQC-I & $~61~$ & $~3~$ & $~61~$ & $~2~$ & $~5~$ & $~4~$ & $~183~$ & $~3,721~$ & $~70~$\\
RQC-II & $~83~$ & $~3~$ & $~83~$ & $~2~$ & $~6~$ & $~4~$ & $~249~$ & $~6,889~$ & $~115~$\\
RQC-III & $~61~$ & $~3~$ & $~61~$ & $~4~$ & $~5~$ & $~4~$ & $~366~$ & $~7,442~$ & $~132~$\\
RQC-IV & $~89~$ & $~5~$ & $~89~$ & $~3~$ & $~6~$ & $~6~$ & $~705~$ & $~12,555~$ & $~192~$\\
\bottomrule \\ % TODO : Once done, check whether BCH(1023, ..., ...) improve n
\end{tabular}
\caption{\label{tab:params_r_post}Parameter sets for quantum-safe RQC, with respect to~\cite{GHT16}. Parameters are analog to Tab.~\ref{tab:params_r_pre}.}
\end{center}
\end{table}
\end{center}

%\begin{center}
%\begin{table}
%\begin{center}
%\begin{tabular}{ccccccccccc}  
%\toprule
%& \multicolumn{9}{c}{Cryptosystem Parameters} \\
%\cmidrule(r){3-11}
%& ~~~~Instance~~~~ & $~n~$ & $~k~$ & $~m~$ & $~q~$ & $~w~$ & $~\epsilon~$ & ~plaintext~ & ~key size~ & ~security~ \\
%\midrule
%\multirow{3}{*}{~\rotatebox{90}{Classical}~~~}&\vspace{.06cm}Medium-I & $~74~$ & $~37~$ & $~37~$ & $~4~$ & $~3~$ & $~3~$ & $~111~$ & $~2,738~$ & $~90~$\\
%&\vspace{.06cm}Medium-II & $~106~$ & $~53~$ & $~53~$ & $~2~$ & $~4~$ & $~4~$ & $~689~$ & $~2,809~$ & $~95~$\\
%&\vspace{.06cm}Strong & $~122~$ & $~61~$ & $~61~$ & $~2~$ & $~5~$ & $~4~$ & $~183~$ & $~3,721~$ & $~140~$\\
%\midrule
%%\multirow{4}{*}{~\rotatebox{90}{Quantum}~~~}&Toy & $~511~$ & $~23~$ & $~11,777~$ & $~67~$ & $~87~$ & $~58~$ & $~174~$ & $~64~$ & $<2^{-30}$\\
%\multirow{3}{*}{~\rotatebox{90}{Quantum}~~~}&\vspace{.06cm}Low & $~122~$ & $~61~$ & $~61~$ & $~2~$ & $~5~$ & $~4~$ & $~183~$ & $~3,721~$ & $~70~$\\
%&\vspace{.06cm}Medium & $~166~$ & $~83~$ & $~83~$ & $~2~$ & $~6~$ & $~4~$ & $~249~$ & $~6,889~$ & $~115~$\\
%&\vspace{.06cm}Strong & $~122~$ & $~61~$ & $~61~$ & $~4~$ & $~5~$ & $~4~$ & $~366~$ & $~7,442~$ & $~132~$\\
%\bottomrule \\ % TODO : Once done, check whether BCH(1023, ..., ...) improve n
%\end{tabular}
%\caption{\label{tab:paramsRank}Parameter sets for our cryptosystem in Rank metric. The plaintexts, key sizes, and security are expressed in bits. For the last three instances, the security is against quantum attacks and expressed in qubits.}
%\end{center}
%\end{table}
%\end{center}

\smallskip\noindent{\bf Computational Cost. } The encryption cost corresponds to a matrix-vector product over $\mathbb{F}_{q^m}$, for a multiplication cost
of elements of $\mathbb{F}_{q^m}$ in $m\log(m)\log(\log(m))$, we obtain an encryption complexity in $\mathcal{O}\left(n^2m\log\left(m\right)\log\left(\log\left(m\right)\right)\right)$.
The decryption cost is also a matrix-vector multiplication plus the decoding cost of the Gabidulin codes, both have the complexities
in $\mathcal{O}\left(n^2m\log\left(m\right)\log\left(\log\left(m\right)\right)\right)$.

\subsection{Comparison with Other Code-based Cryptosystems}

In the following we consider the different types of code-based cryptosystems and express different parameters
of the different systems in terms of the security parameters $\lambda$, considering best known attacks of complexity
$2^{\mathcal{O}(w)}$ for decoding a word of weight $w$ for Hamming distance and complexity in $2^{\mathcal{O}(wn)}$
for decoding a word of rank weight $w$ for a code of double-circulant code of length $2n$ for rank metric.
McEliece-Goppa corresponds to the original scheme proposed by McEliece~\cite{ME78} of dimension rate $\frac{1}{2}$. 

\begin{table}
\begin{center}
\begin{tabular}{rlccccc}
\toprule
\multicolumn{2}{c}{\multirow{2}{*}{Cryptosystem}} & Code & Public & Ciphertext & Hidden & Cyclic \\
& & Length & Key Size & Size & Structure & Structure \\
\midrule
Goppa- & \multirow{2}{*}{\!\!\cite{ME78}} & \multirow{2}{*}{$\mathcal{O}\left(\lambda\log\lambda\right)$} & \multirow{2}{*}{$\mathcal{O}\left(\lambda^2\left(\log\lambda\right)^2\right)$} & \multirow{2}{*}{$\mathcal{O}\left(\lambda\log\lambda\right)$} & \multirow{2}{*}{Strong} & \multirow{2}{*}{No} \\
McEliece&&&&&&\\[1.33mm]
MDPC & \cite{MTSB13} &  $\mathcal{O}\left({\lambda}^2\right)$ & $\mathcal{O}\left({\lambda}^2\right)$ & $\mathcal{O}\left({\lambda}^2\right)$ & Weak & Yes\\[1.33mm]
LRPC & \cite{GMRZ13} & $\mathcal{O}\left({\lambda}^{\frac{2}{3}}\right)$ & $\mathcal{O}\left({\lambda}^{\frac{4}{3}}\right)$ & $\mathcal{O}\left({\lambda}^{\frac{4}{3}}\right)$ 
& Weak & Yes\\[1.33mm]
HQC & [Sec.~\ref{sec:7-1-hamming}] &  $\mathcal{O}\left({\lambda}^2\right)$ & $\mathcal{O}\left({\lambda}^2\right)$ & $\mathcal{O}\left({\lambda}^2\right)$ & No & Yes\\[1.33mm]
RQC & [Sec.~\ref{sec:7-2-rank}] &  $\mathcal{O}\left({\lambda}^{\frac{2}{3}}\right)$ & $\mathcal{O}\left({\lambda}^{\frac{4}{3}}\right)$ & $\mathcal{O}\left({\lambda}^{\frac{4}{3}}\right)$  &
No & Yes\\
\bottomrule
\end{tabular}
\end{center}
\caption{\label{tab:paramsComp}Parameters comparison for different code-based cryptosystems with respect to the security parameter $\lambda$}
\end{table}

Tab.~\ref{tab:paramsComp} shows that even if the recent cryptosystem MDPC has a smaller public key
and a weaker hidden structure than the McEliece cryptosystem, the size of the ciphertext remains non negligible.
The HQC benefits from the same type of parameters than the MDPC systems but with no hidden structure at the cost
of a smaller encryption rate. Finally, the table shows the very strong potential of rank metric based cryptosystems,
whose parameters remain rather low compared to MDPC and HQC cryptosystems.  

\section{Conclusion and Future Work}
\label{sec:conclusion}

We have presented an efficient approach for constructing 
code-based cryptosystems. This approach originates in Alekhnovich's blueprint~\cite{FOCS:Alekhnovich03} on random 
matrices. Our construction is generic enough so 
that we provide two instantiations of our cryptosystem: one 
for the Hamming metric (HQC), and one for the Rank metric (RQC). Both 
constructions are pretty efficient and compare favourably 
to previous work, especially for the rank metric setting. Additionally, we provide for the Hamming 
setting an analysis of the error term yielding a 
concrete, precise and easy-to-verify decryption failure.

This analysis was facilitated by the shape 
of the tensor product code, and more complex-to-analyze 
tensor product codes might yield slightly shorter keys and better
efficiency.

However, for such a tensor product code the analysis of the 
decryption failure probability becomes much more tricky, and 
finding suitable upper bounds for it will involve future work.

%\section{Application to Key Exchange ?}
%\label{sec:keyExchange}
%\philippe{Il est probable que cette partie saute...}
%
%\jc{\at{Philippe} je pense que si on doit utiliser la notation $(\mathbf{a, \alpha})/(\mathbf{b, \beta})$ c'est dans cette partie. Sinon le fait qu'on ne puisse pas mettre les lettres grecques en gras va nous faire perdre en compréhension (par rapport à la notation vectorielle). Qu'en penses-tu ?}
%
%\jc{Finalement à voir avec \textbackslash\texttt{bm\{\}}, cf mail Olivier...}
%
%\jc{LPN $\rightarrow$ \cite{C:Stern93} application to public-key
%authentication}

%{\color{red} Références à checker (manque certains noms d'ouvrages et/ou de confs...}

\bibliographystyle{alpha}

\newcommand{\etalchar}[1]{$^{#1}$}

%\bibliography{./cryptobib/abbrev3,./cryptobib/crypto_crossref,./externbib/other}
%\appendix
%\input{appendix-goldreich-leivin}
%\input{appendix-proba}
%\input{appendix-probaTotal}

%\input{background}

%\section*{Philippe's missing references $\rightarrow$ TO CHECK !}
%
%\begin{itemize}
%\item BCC07~\cite{BCC07}
%\item HT15~\cite{HT15}
%\item SKK10~\cite{SKK10}
%%\item doom : \cite{doom}
%%\item DCC : \cite{DCC}
%%\item CRYPTO2007 : \cite{CRYPTO2007}
%%\item ????????? : \cite{?????????}
%%\item CS96 : \cite{CS96}
%%\item OJ02 : \cite{OJ02}
%%\item GRS12 : \cite{GRS12}
%%\item  : \cite{}
%%\item issac10 : \cite{issac10}
%%\item Loi06 : \cite{Loi06}
%%\item Gab85 : \cite{Gab85}
%%\item Rosa?? : \cite{Rosa??}
%\end{itemize}

\end{document}